\documentclass{irmaart}
\usepackage[pagebackref,final]{hyperref}
\usepackage{lineno,ha,pgf,tikz,bbold,color}
\usetikzlibrary{arrows,automata}
\usepackage{people}

\newif\ifnazi
\nazifalse

\newcommand{\Card}[1]{\ifnazi\operatorname{Card}#1\else{\##1}\fi}
\newcommand{\rk}{\operatorname{rk}}
\newcommand{\Rat}{\operatorname{Rat}}
\newcommand{\Rec}{\operatorname{Rec}}
\newcommand{\cl}{\operatorname{cl}}
\newcommand{\one}{{\mathbb 1}}
\newcommand\GL{{\operatorname{GL}}}

\newcommand{\mapright}[1]{\smash{\mathop{\longrightarrow}\limits^{#1}}}
\newcommand{\mapleft}[1]{\smash{\mathop{\longleftarrow}\limits^{#1}}}

\newcommand{\vlongrightarrow}{\relbar\joinrel\longrightarrow}
\newcommand{\vvlongrightarrow}{\relbar\joinrel\vlongrightarrow}
\newcommand{\vvvlongrightarrow}{\relbar\joinrel\vvlongrightarrow}
\newcommand{\vvvvlongrightarrow}{\relbar\joinrel\vvvlongrightarrow}
\newcommand{\longmapright}[1]{\smash{\mathop{\vlongrightarrow}\limits^{#1}}}
\newcommand{\vlongmapright}[1]{\smash{\mathop{\vvlongrightarrow}\limits^{#1}}}

\newcommand{\vvvlongmapright}[1]{\smash{\mathop{\vvvvlongrightarrow}\limits^{#1}}}
\newcommand\chapterref[1]{23.\ref{#1}}

\newenvironment{fsa}{\begin{tikzpicture}[->,>=stealth',
    shorten >=1pt,auto,node distance=3cm,
    initial text=,
semithick]}{\end{tikzpicture}}

\begin{document}

\title{Rational subsets of groups}

\author{L.~Bartholdi$^{1}$ \and P.~V.~Silva$^{2,}$\thanks{The second
    author acknowledges support by Project ASA (PTDC/MAT/65481/2006)
    and C.M.U.P., financed by F.C.T. (Portugal) through the programmes
    POCTI and POSI, with national and E.U. structural funds.}}

\date{2010-12-07}

\markboth{L.~Bartholdi, P.~V.~Silva}{Rational subsets of groups}

\address{ $^1$Mathematisches Institut\\
  Georg-August Universit\"at zu G\"ottingen\\
  Bunsenstra\ss e 3--5\\
  D-37073 G\"ottingen, Germany\\
  email:\,\url{laurent.bartholdi@gmail.com}
  \\[4pt]
  $^2$Centro de Matem\'{a}tica, Faculdade de Ci\^{e}ncias\\
  Universidade do Porto\\
  R. Campo Alegre 687\\
  4169-007 Porto, Portugal\\
  email:\,\url{pvsilva@fc.up.pt} }

\maketitle\label{chapterBS1}

\vspace{-8mm}



\vspace{-1cm}
\localtableofcontents

Over the years, finite automata have been used effectively in the
theory of infinite groups to represent rational subsets. This includes
the important particular case of finitely generated subgroups (and the
beautiful theory of Stallings automata for the free group case), but
goes far beyond that: certain inductive procedures need a more general
setting than mere subgroups, and rational subsets constitute the
natural generalization.  The connections between automata theory and
group theory are rich and deep, and many are portrayed in Sims'
book~\cite{MR1267733}.

This chapter is divided into three parts: in Section~\ref{BS1:fg} we introduce
basic concepts, terminology and notation for finitely generated
groups, devoting special attention to free groups. These will also be
used in Chapter~\ref{chapterBS2}.

Section~\ref{BS1:inverse} describes the use of finite inverse automata to
study finitely generated subgroups of free groups. The automaton
recognizes elements of a subgroup, represented as words in the ambient
free group.

Section~\ref{BS1:rational} considers, more generally, rational subsets of
groups, when good closure and decidability properties of these subsets
are satisfied.

The authors are grateful to Stuart Margolis, Benjamin Steinberg and
Pascal Weil for their remarks on a preliminary version of this text.


\section{Finitely generated groups}\label{BS1:fg}

Let $G$ be a group.  Given $A \subseteq G$, let $\langle A \rangle =
(A \cup A^{-1})^*$ denote the subgroup of $G$ \emph{generated} by
$A$. We say that $H \leq G$ is \emph{finitely
  generated}\index{subgroup!finitely~generated} and write $H
\leq_{f.g.} G$ if $H = \langle A \rangle$ for some finite subset $A$
of $H$.

Given $H \leq G$, we denote by $[G:H]$ the
\emph{index}\index{subgroup!index~of~a} of $H$ in $G$, that is, the
number of right cosets $Hg$ for all $g \in G$; or, equivalently, the
number of left cosets. If $[G:H]$ is finite, we write $H \leq_{f.i.}
G$. It is well known that every finite index subgroup of a finitely
generated group is finitely generated.

We denote by $\one$ the identity of $G$.  An element $g \in G$ has
\emph{finite order}\index{finite~order~element} if $\langle g \rangle$
is finite. Elements $g,h \in G$ are \emph{conjugate}\index{conjugate~elements} if $h = x^{-1}gx$ for some $x \in G$. We use the notation
$g^h = h^{-1}gh$ and $[g,h] = g^{-1}g^h$ to denote, respectively,
conjugates and commutators\index{commutator~in~groups}.

Given an alphabet $A$, we denote by $A^{-1}$ a set of \emph{formal
  inverses} of $A$,
and write $\widetilde{A} = A\cup A^{-1}$. We say that $\widetilde{A}$
is an \emph{involutive alphabet}\index{alphabet!involutive}.
We extend $^{-1}:A \to A^{-1}: a \mapsto a^{-1}$ to an involution on
$\widetilde{A}^*$ through
\begin{displaymath}
(a^{-1})^{-1} = a,\quad (uv)^{-1} = v^{-1}u^{-1}\quad (a \in A,\;
u,v \in \widetilde{A}^*)\, .
\end{displaymath}

\index{problem,~decision|(} If $G = \langle A \rangle$, we have a
canonical epimorphism $\rho:\widetilde{A}^* \twoheadrightarrow G$,
mapping $a^{\pm1}\in\widetilde A$ to $a^{\pm1}\in G$. We present next
some classical decidability problems:

\begin{definition}
\label{BS1:decproblems}
  Let $G=\langle A\rangle$ be a finitely generated group.
  \begin{description}
  \item[word problem:]\index{word~problem}\index{problem,~decision!word}
    is there an algorithm that, upon receiving as input a word $u\in
    \widetilde{A}^*$, determines whether or not $\rho(u) =\one$?
  \item[conjugacy problem:]\index{conjugacy~problem}\index{problem,~decision!conjugacy}

    is there an algorithm that, upon receiving as input words $u,v \in
    \widetilde{A}^*$, determines whether or not $\rho(u)$ and
    $\rho(v)$ are conjugate in $G$?
  \item[membership problem for ${\cal{K}} \subseteq 2^G$:]\index{membership~problem}\index{problem,~decision!membership}
    is there for every $X\in\cal K$ an algorithm that, upon receiving
    as input a word $u \in \widetilde{A}^*$, determines whether or not
    $\rho(u) \in X$?
  \item[generalized word problem:]\index{generalized~word~problem}\index{problem,~decision!word!generalized}
    is the membership problem for the class of finitely generated
    subgroups of $G$ solvable?
  \item[order problem:]\index{order~problem}\index{problem,~decision!order}
    is there an algorithm that, upon receiving as input a word $u \in
    \widetilde{A}^*$, determines whether $\rho(u)$ has finite or
    infinite order?
  \item[isomorphism problem for a class ${\cal{G}}$ of groups:]\index{isomorphism~problem}\index{problem,~decision!isomorphism}  
    is there an algorithm that, upon receiving as input a description
    of groups $G,H \in {\cal{G}}$, decides whether or not $G\cong H$?

    Typically, $\cal G$ may be a subclass of finitely presented
    groups\index{group!finitely~presented} (given by their
    presentation), or automata
    groups\index{automata~group}\index{group!automata} (see
    Chapter~\ref{chapterBS2}) given by automata.
  \end{description}
\end{definition}

We can also require complexity bounds on the algorithms; more
precisely, we may ask with which complexity bound an answer to the
problem may be obtained, and also with which complexity bound a
witness (a normal form for the word problem, an element conjugating
$\rho(u)$ to $\rho(v)$ in case they are conjugate, an expression of
$u$ in the generators of $X$ in the generalized word problem) may be
constructed.  \index{problem,~decision|)}

\subsection{Free groups}
We recall that an equivalence relation $\sim$ on a semigroup $S$ is
a \emph{congruence}\index{congruence} if $a\sim b$ implies
$ac\sim bc$ and $ca\sim cb$ for all $a,b,c \in S$.

\begin{definition}
\label{BS1:freegroup}
Given an alphabet $A$, let $\sim$ denote the congruence on
$\widetilde{A}^*$ generated by the relation
\begin{equation}
\label{BS1:relfg}
\{(aa^{-1},1)\mid a \in \widetilde{A}\}\, .
\end{equation}
The quotient $F_A = \widetilde{A}^*/{\sim}$ is the \emph{free group
  on}\index{free~group}\index{group!free}
$A$. We denote by $\theta: \widetilde{A}^* \to F_A$ the canonical
morphism $u \mapsto [u]_{\sim}$.
\end{definition}

Free groups admit the following universal property: for every map
$f:A\to G$, there is a unique group morphism $F_A\to G$ that extends
$f$.

Alternatively, we can view~\eqref{BS1:relfg} as a {\em
  confluent}\index{rewriting~system!confluent}
length-reducing rewriting system on $\widetilde{A}^*$, where each word
$w \in \widetilde{A}^*$ can be transformed into a unique
\emph{reduced}\index{word!reduced} word $\overline{w}$
with no factor of the form $aa^{-1}$, see~\cite{Book&Otto:1993}.  As a
consequence, the equivalence
\begin{displaymath}
u\sim v \hspace{.5cm} \Leftrightarrow \hspace{.5cm}\overline{u} = \overline{v}
\hspace{2cm} (u,v \in \widetilde{A}^*)
\end{displaymath}
solves the word problem for $F_A$.

We shall use the notation
$R_A = \overline{\widetilde{A}^*}$. It is well known that $F_A$ is
isomorphic to $R_A$ under the binary operation
\begin{displaymath}
u\star v = \overline{uv} \hspace{1cm} (u,v \in R_A)\, .
\end{displaymath}
We recall that the \emph{length} $|g|$ of $g\in F_A$ is the
length of the reduced form of $g$, also denoted by $\overline{g}$.

The letters of $A$ provide a natural \emph{basis}\index{free~group!basis} for $F_A$: they generate $F_A$ and satisfy no
nontrivial relations, that is, all reduced words on these generators
represent distinct elements of $F_A$. A group is free if and only if
it has a basis.

Throughout this chapter, we assume $A$ to be a finite alphabet. It is
well known that free groups $F_A$ and $F_B$ are isomorphic if and only
if $\Card A = \Card B$. This leads to the concept of
\emph{rank}\index{free~group!rank} of a free group $F$: the
\emph{cardinality} of a basis of $F$, denoted by $\rk F$. It is common to
use the notation $F_n$ to denote a free group of rank $n$.

We recall that a reduced word $u$ is \emph{cyclically
  reduced}\index{word!cyclically~reduced} if $uu$
is also reduced. Any reduced word $u\in R_A$ admits a unique
decomposition of the form $u = vwv^{-1}$ with $w$ cyclically
reduced. A solution for the conjugacy problem follows easily from
this: first reduce the words cyclically; then two cyclically reduced
words in $R_A$ are conjugate if and only if they are cyclic
permutations of each other. On the other hand, the order problem
admits a trivial solution: only the identity has finite
order. Finally, the generalized word problem shall be discussed in the
following section.

\section{Inverse automata and Stallings' construction}\label{BS1:inverse}

The study of finitely generated subgroups of free groups entered a new
era in the early eighties when \Stallings\ made explicit and effective
a construction~\cite{MR695906} that can be traced back to the early
part of the twentieth century in \Schreier's coset graphs
(see~\cite{MR1267733} and~\S\ref{BS2:sec:1}) and to \Serre's
work~\cite{MR0476875}. Stallings' seminal paper was built over
\emph{immersions of finite graphs}, but the alternative approach using
finite inverse automata became much more popular over the years; for
more on their link, see~\cite{MR1214007}. An extensive survey has been
written by \Kapovich\ and \Miasnikov~\cite{MR1882114}.

Stallings' construction for $H \leq_{f.g.} F_A$
consists in taking a finite set of generators for $H$ in reduced form,
building the so-called \emph{flower automaton} and then proceeding to
make this automaton deterministic through the operation known as
\emph{Stallings foldings}. This is clearly a terminating procedure,
but the key fact is that the construction is independent from both the
given finite generating set and the chosen folding sequence. A short
simple automata-theoretic proof of this claim will be given. The finite inverse
automaton ${\cal{S}}(H)$ thus obtained is usually called the
\emph{Stallings automaton} of $H$. Over the years, Stallings automata
became the standard representation for finitely generated subgroups of
free groups
and are involved in many of the algorithmic results presently
obtained.

Several of these algorithms are implemented in computer software, see
e.g.\ \textsc{CRAG}~\cite{CRAG}, or the packages \textsc{Automata} and
\textsc{FGA} in~\textsc{GAP}~\cite{TheGAPGroup:2004}.

\subsection{Inverse automata}

An automaton $\cal{A}$ over an involutive alphabet $\widetilde{A}$ is
\emph{involutive}\index{automaton!involutive} if, whenever $(p,a,q)$
is an edge of $\cal{A}$, so
is $(q,a^{-1},p)$.  Therefore it suffices to depict just the
\emph{positively labelled} edges (having label in $A$) in their
graphical representation.

\begin{definition}
\label{BS1:inverseautomaton}
An involutive automaton is \emph{inverse}\index{automaton!inverse} if
it is deterministic, trim and has a
single final state.
\end{definition}

If the latter happens to be the initial state, it is called the
\emph{basepoint}. It follows easily from the computation of the
\Nerode\ equivalence (see~\S10.2) that every inverse automaton is a
minimal automaton.

Finite inverse automata capture the idea of an action (of a finite
\emph{inverse monoid}, their \emph{transition
  monoid}\index{transition~monoid}\index{monoid!transition}) on a
finite set (the vertex set) through partial bijections. We recall that
a monoid $M$ is inverse\index{monoid!inverse} if, for every $x \in M$,
there exists a unique $y \in M$ such that $xyx = x$ and $y=yxy$; then
$M$ acts by partial bijections on itself.

\noindent The next result is easily proven, but is quite useful.

\begin{proposition}
\label{BS1:invaut}
Let $\cal{A}$ be an inverse automaton and let $p
\vvvlongmapright{uvv^{-1}w} q$ be a path in $\cal{A}$. Then there
exists also a path $p \mapright{uw} q$ in $\cal{A}$.
\end{proposition}

Another important property relates languages to morphisms. For us, a
\emph{morphism}\index{morphism of deterministic automata} between
deterministic automata $\cal{A}$ and $\cal{A}'$ is a mapping $\varphi$
between their respective vertex sets which preserves initial vertices,
final vertices and edges, in the sense that
$(\varphi(p),a,\varphi(q))$ is an edge of $\cal{A}'$ whenever
$(p,a,q)$ is an edge of $\cal{A}$.

\begin{proposition}
\label{BS1:invmor}
Given inverse automata $\cal{A}$ and $\cal{A}'$, then $L({\cal{A}})
\subseteq L({\cal{A}}')$ if and only if there exists a morphism
$\varphi: \cal{A} \to \cal{A}'$. Moreover, such a morphism is
unique.\index{morphism of deterministic automata} 
\end{proposition}

\begin{proof}
$(\Rightarrow)$: Given a vertex $q$ of
$\cal{A}$, take a successful path
\begin{displaymath}
\to q_0 \mapright{u} q \mapright{v} t \to
\end{displaymath}
in $\cal{A}$, for some $u,v\in \widetilde A^*$. Since $L({\cal{A}}) \subseteq
L({\cal{A}}')$, there exists a successful path
\begin{displaymath}
\to q'_0 \mapright{u} q' \mapright{v} t' \to
\end{displaymath}
in $\cal{A}'$. We take $\varphi(q) = q'$.

To show that $\varphi$ is well defined, suppose that
\begin{displaymath}
\to q_0 \mapright{u'} q \mapright{v'} t \to
\end{displaymath}
is an alternative successful path in $\cal{A}$. Since $u'v \in
L({\cal{A}}) \subseteq
L({\cal{A}}')$, there exists a successful path
\begin{displaymath}
\to q'_0 \mapright{u'} q'' \mapright{v} t' \to
\end{displaymath}
in $\cal{A}'$ and it follows that $q' = q''$ since ${\cal{A}}'$ is
inverse. Thus $\varphi$ is well defined.

It is now routine to check that
$\varphi$ is a morphism from $\cal{A}$ to $\cal{A}'$ and
that it is unique.

$(\Leftarrow)$: Immediate from the definition of morphism.
\end{proof}

\subsection{Stallings' construction}\index{Stallings~construction|(}

Let $X$ be a finite subset of $R_A$. We build an involutive automaton
${\cal{F}}(X)$ by fixing a basepoint $q_0$ and
gluing to it a \emph{petal} labelled by every word in $X$ as follows:
if $x = a_1\dots a_k \in X$, with
$a_i \in \widetilde{A}$, the petal consists of a closed path of the form
\begin{displaymath}
q_0 \mapright{a_1} \bullet \mapright{a_2} \cdots \mapright{a_k} q_0
\end{displaymath}
and the respective inverse edges. All such intermediate vertices
$\bullet$ are assumed to be distinct in the automaton. For obvious
reasons, ${\cal{F}}(X)$ is called the \emph{flower
  automaton}\index{automaton!flower} of $X$.

The automaton ${\cal{F}}(X)$ is almost an inverse automaton -- except
that it need 
not be deterministic. We can fix it by performing a sequence of
so-called \emph{Stallings foldings}. Assume that ${\cal{A}}$ is a trim
involutive automaton with a basepoint, possessing two distinct edges
of the form
\begin{equation}
\label{BS1:fold}
p \mapright{a} q, \quad p \mapright{a} r
\end{equation}
for $a \in \widetilde{A}$.
The \emph{folding} is performed by identifying these two edges, as well
as the two respective inverse edges.
In particular, the vertices $q$ and $r$ are also identified (if they
were distinct).

The number of edges is certain to decrease through foldings. Therefore,
if we perform enough of them,
we are sure to turn ${\cal{F}}(X)$ into a finite inverse automaton.
\begin{definition}
\label{BS1:stallingsautomaton}
The \emph{Stallings automaton}\index{automaton!Stallings}\index{Stallings~automaton} of $X$ is the finite inverse automaton
${\cal{S}}(X)$ obtained through folding ${\cal F}(X)$.
\end{definition}

We shall see that
${\cal{S}}(X)$ depends only on the finitely generated subgroup
$\langle X \rangle$ of $F_A$ generated by $X$, being in particular
independent from the choice
of foldings taken to reach it.

Since inverse automata are minimal, it suffices to characterize
$L({\cal{S}}(X))$ in terms of $H$ to prove
uniqueness (up to isomorphism):

\begin{proposition}
\label{BS1:uniqstal}
Fix $H \leq_{f.g.} F_A$ and let $X \subseteq R_A$ be a finite
generating set for $H$.
Then
\begin{multline*}
L({\cal{S}}(X)) = \bigcap\{L \subseteq \widetilde{A}^* \mid L
\mbox{ is recognized by a finite inverse
automaton}\\
\mbox{with a basepoint and }
\overline{H} \subseteq L\} \, .
\end{multline*}
\end{proposition}

\begin{proof}



$(\supseteq)$: Clearly, ${\cal{S}}(X)$ is a finite inverse
automaton with a basepoint. Since $X\cup X^{-1} \subseteq L({\cal{F}}(X))
\subseteq L({\cal{S}}(X))$, it follows easily
from Proposition~\ref{BS1:invaut} that
\begin{equation}
\label{BS1:uniqstal1}
\overline{H} \subseteq L({\cal{S}}(X))\, .
\end{equation}

$(\subseteq)$: Let $L \subseteq \widetilde{A}^*$ be recognized by a
finite inverse
automaton $\cal{A}$ with a basepoint, with $\overline{H} \subseteq
L$. Since $X \subseteq \overline{H}$, we have an automaton morphism
from ${\cal{F}}(X)$ to $\cal{A}$, hence $L({\cal{F}}(X)) \subseteq
L$. To prove that $L({\cal{S}}(X)) \subseteq L$,
it suffices to show that inclusion in $L$ is preserved through foldings.

Indeed, assume that $L({\cal{B}}) \subseteq L$ and $\cal{B}'$ is
obtained from $\cal{B}$ by folding the two edges
in~\eqref{BS1:fold}. It is immediate that every successful path $q_0
\mapright{u} t$ in $\cal{B}'$ can be lifted to a
successful path $q_0 \mapright{v} t$ in $\cal{B}$ by successively
inserting the word $a^{-1}a$ into $u$. Now
$v \in L = L(\cal{A})$ implies $u \in L$ in view of Proposition~\ref{BS1:invaut}.
\end{proof}

Now, given $H\leq F_A$ finitely generated, we take a finite set $X$
of generators. Without loss of generality, we may assume that $X$ consists
of reduced words, and
we may define ${\cal{S}}(H) = {\cal{S}}(X)$ to be the \emph{Stallings
  automaton} of $H$.

\begin{example}
\label{BS1:exstal}
Stallings' construction for $X = \{ a^{-1}ba, ba^2 \}$, where the next
edges to be identified are depicted by dotted lines, is
\end{example}
$$\begin{fsa}
\node[state] (q_0) at (0,1) {$q_0$};
\node[state] (E) at (2,0) {};
\node[state] (NE) at (2,2) {};
\node[state] (W) at (-2,0) {};
\node[state] (NW) at (-2,2) {};
\node at (-4,1) {${\cal{F}}(X) =$};
\path (q_0) edge node[below] {$b$} (E);
\path (W) edge node[below] {$a$} (q_0) edge node {$b$} (NW);
\path (NW) edge[dotted] node {$a$} (q_0);
\path (NE) edge[dotted] node[above] {$a$} (q_0);
\path (E) edge node[right] {$a$} (NE);
\end{fsa}
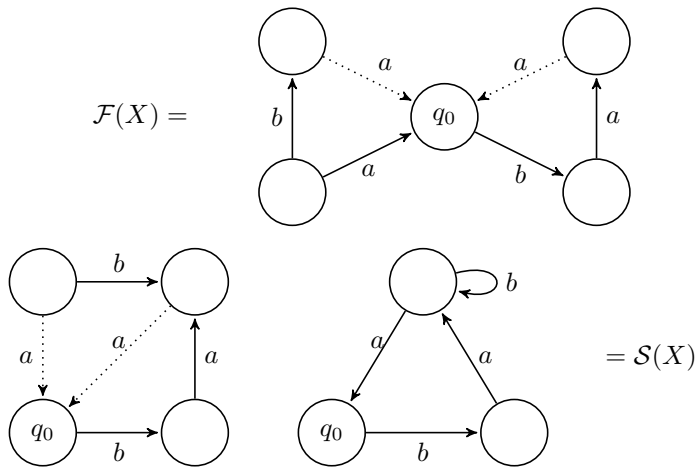
$$
$$\begin{fsa}
\node[state] (q_0) at (-3,0) {$q_0$};
\node[state] (E) at (-1,0) {};
\node[state] (N) at (-3,2) {};
\node[state] (NE) at (-1,2) {};
\node[state] (p_0) at (0.8,0) {$q_0$};
\node[state] (EE) at (3.2,0) {};
\node[state] (NNEE) at (2,2) {};
\node at (5,1) {$= {\cal{S}}(X)$};
\path (q_0) edge node[below] {$b$} (E);
\path (N) edge node {$b$} (NE) edge[dotted] node[left] {$a$} (q_0);
\path (NE) edge[dotted] node[above] {$a$} (q_0);
\path (E) edge node[right] {$a$} (NE);
\path (p_0) edge node[below] {$b$} (EE);
\path (NNEE) edge node[above] {$a$} (p_0) edge[loop right] node {$b$} ();
\path (EE) edge node[right] {$a$} (NNEE);
\end{fsa}$$

A simple, yet important example is given by applying the construction
to $F_n$ itself, when we obtain the so-called \emph{bouquet} of
$n$ circles:
$$\begin{fsa}
\node[state] (1) at (-3,0) {$q_0$};
\node[state] (2) at (0,0) {$q_0$};
\node[state] (3) at (3,0) {$q_0$};
\node at (-3,-1) {${\cal{S}}(F_1)$};
\node at (0,-1) {${\cal{S}}(F_2)$};
\node at (3,-1) {${\cal{S}}(F_3)$};
\path (1) edge[loop left] node {$a$} ();
\path (2) edge[loop left] node {$a$} () edge[loop right] node {$b$} ();
\path (3) edge[loop left] node {$a$} () edge[loop right] node {$b$} ()
edge[loop above] node {$c$} ();
\end{fsa}$$

In terms of complexity, the best known algorithm for the construction
of ${\cal{S}}(X)$ is due to \Touikan~\cite{Touikan:2006}. Its
time complexity is $O(n\log^*n)$, where $n$ is the sum of the lengths of
the elements of $X$.

\subsection{Basic applications}
The most fundamental application of Stallings' construction is an
elegant and efficient solution to the generalized word problem:
\begin{theorem}
\label{BS1:thm:wordpb}
  The generalized word problem\index{free~group!generalized~word~problem} in $F_A$ is solvable.
\end{theorem}

We will see many groups in Chapter~\ref{chapterBS2} that have solvable
word problem; however, few of them have solvable generalized word
problem. The proof of Theorem~\ref{BS1:thm:wordpb} relies on
\begin{proposition}
\label{BS1:gwp}
Consider $H \leq_{f.g.} F_A$ and $u \in F_A$. Then $u \in H$ if and only if
$\overline{u} \in L({\cal{S}}(H))$.
\end{proposition}

\begin{proof}
$(\Rightarrow)$: Follows from~\eqref{BS1:uniqstal1}.

$(\Leftarrow)$: It follows easily from the last paragraph of the proof
of Proposition~\ref{BS1:uniqstal} that, if $\cal{B}'$ is obtained from
$\cal{B}$ by performing Stallings foldings, then
$\overline{L(\cal{B}')} = \overline{L(\cal{B})}$. Hence, if $H =
\langle X\rangle$, we get
\begin{displaymath}
\overline{L({\cal{S}}(H))} = \overline{L({\cal{F}}(X))} = \overline{(X
  \cup X^{-1})^*} = \overline{H}
\end{displaymath}
and the implication follows.
\end{proof}

It follows from our previous remark that the complexity of the
generalized word problem is $O(n\log^*n+m)$, where $n$ is the sum of
the lengths of the elements of $X$ and $m$ is the length of the input
word. In particular, once the subgroup $X$ has been fixed, complexity
is linear in $m$.

\begin{example}
\label{BS1:exstal1}
We may use the Stallings automaton constructed in Example~\ref{BS1:exstal} to check that $baba^{-1}b^{-1} \in H = \langle a^{-1}ba,
ba^2 \rangle$ but $ab \notin H$.
\end{example}

Stallings automata also provide an effective construction for bases
of finitely generated subgroups. Consider $H \leq_{f.g.} F_A$, and let
$m$ be the number of vertices of ${\cal{S}}(H)$. A \emph{spanning
  tree}\index{spanning~tree} $T$ for ${\cal{S}}(H)$ consists of $m-1$ edges and their
inverses which, together, connect all the vertices of ${\cal{S}}(H)$.
Given a vertex $p$ of ${\cal{S}}(H)$, we denote by $g_p$ the
$T$-\emph{geodesic} connecting the basepoint $q_0$ to $p$, that is,
$q_0 \mapright{g_p} p$ is the shortest path contained in $T$
connecting $q_0$ to $p$.

\begin{proposition}
\label{BS1:basis}
Let $H \leq_{f.g.} F_A$ and let $T$ be a spanning tree for
${\cal{S}}(H)$. Let $E_+$ be the set of
positively labelled edges of ${\cal{S}}(H)$. Then $H$ is free with basis
\begin{displaymath}
Y = \{ g_pag_q^{-1} \mid (p,a,q) \in E_+\setminus T \}\, .
\end{displaymath}
\end{proposition}

\begin{proof}
It follows from Proposition~\ref{BS1:gwp} that $L({\cal{S}}(H)) \subseteq H$,
hence $Y \subseteq H$. To show that
$H = \langle Y\rangle$, take $h = a_1\cdots a_k \in H$ in reduced form
$(a_i \in \widetilde{A})$. By
Proposition~\ref{BS1:gwp}, there exists a successful path
\begin{displaymath}
q_0 \mapright{a_1} q_1 \mapright{a_2} \cdots \mapright{a_k} q_k = q_0
\end{displaymath}
in ${\cal{S}}(H)$. For $i = 1,\dots,k$, we have either
$g_{q_{i-1}}a_ig_{q_i}^{-1} \in Y \cup Y^{-1}$ or
$\overline{g_{q_{i-1}}a_ig_{q_i}^{-1}} = 1$, the latter occurring if $(q_{i-1},a_i,q_i) \in
T$. In any case, we get
\begin{displaymath}
h = a_1\cdots a_k =
\overline{(g_{q_{0}}a_1g_{q_1}^{-1})(g_{q_{1}}a_2g_{q_2}^{-1})
  \cdots(g_{q_{k-1}}a_kg_{q_0}^{-1})}
\in \langle Y \rangle
\end{displaymath}
and so $H = \langle Y\rangle$.

It remains to show that the elements of $Y$ satisfy no nontrivial
relations. Let $y_1,\dots, y_k$ $\in Y\cup Y^{-1}$
with $y_i \neq y_{i-1}^{-1}$ for $i = 2,\dots,k$. Write $y_i =
g_{p_i}a_ig_{r_i}^{-1}$,
where $a_i \in \widetilde{A}$
labels the edge not in $T$. It follows easily from $y_i \neq
y_{i-1}^{-1}$ and the definition of spanning tree that
\begin{displaymath}
\overline{y_1\cdots y_k} = g_{p_1}a_1\overline{g_{r_1}^{-1}g_{p_2}}a_2 \cdots
a_{k-1}\overline{g_{r_{k-1}}^{-1}g_{p_k}}a_kg_{r_k}\, ,
\end{displaymath}
a nonempty reduced word if $k \geq 1$. Therefore $Y$ is a basis of $H$
as claimed.
\end{proof}

\noindent In the process, we also obtain a proof of the
\Nielsen-\Schreier\ Theorem, in the case of finitely generated
subgroups. A simple topological proof may be found
in~\cite{54.0603.01}:

\begin{theorem}[Nielsen-Schreier]\label{BS1:nielsen}
  Every subgroup of a free group is itself free.
\end{theorem}

\begin{example}
\label{BS1:exstal2}
We use the Stallings automaton constructed in Example~\ref{BS1:exstal} to construct a basis of $H = \langle a^{-1}ba,
ba^2 \rangle$.
\end{example}

If we take the spanning tree $T$ defined by the dotted lines in
$$\begin{fsa}
\node[state] (p_0) at (0.8,0) {$q_0$};
\node[state] (EE) at (3.2,0) {};
\node[state] (NNEE) at (2,2) {};
\path (p_0) edge[dotted] node[below] {$b$} (EE);
\path (NNEE) edge node[above] {$a$} (p_0) edge[loop right] node {$b$} ();
\path (EE) edge[dotted] node[right] {$a$} (NNEE);
\end{fsa}$$
then $\Card{E_+\setminus T} = 2$ and the corresponding basis is $\{ ba^2,
baba^{-1}b^{-1} \}$. Another choice of spanning tree actually proves that the
original generating set is also a basis.

\bigskip

We remark that Proposition~\ref{BS1:basis} can be extended to the case
of infinitely generated subgroups, proving the general case of
Theorem~\ref{BS1:nielsen}. However, in this case there is no effective
construction such as Stallings', and the (infinite) inverse automaton
${\cal{S}}(H)$ remains a theoretical object, using appropriate cosets
as vertices.

Another classical application of Stallings' construction regards the
identification of finite index subgroups.

\begin{proposition}
\label{BS1:findex}
Consider $H \leq_{f.g.} F_A$.
\begin{conditionsiii}
\item[\textup{(i)}] $H$ is a finite index subgroup of $F_A$ if and only if
${\cal{S}}(H)$ is a complete automaton.
\item[\textup{(ii)}] If $H$ is a finite index subgroup of $F_A$, then its index
  is the number of vertices of ${\cal{S}}(H)$.
\end{conditionsiii}
\end{proposition}

\begin{proof}
(i) $(\Rightarrow)$: Suppose that ${\cal{S}}(H)$ is not complete. Then
there exist some vertex $q$ and some $a \in \widetilde{A}$ such that
$q\cdot a$ is undefined. Let $g$ be a geodesic connecting the basepoint
$q_0$ to $q$ in ${\cal{S}}(H)$. We claim that
\begin{equation}
\label{BS1:findex1}
Hga^m \neq Hga^n\quad \mbox{if} \quad m-n
> |g|\, .
\end{equation}
Indeed, $Hga^m = Hga^n$ implies $ga^{m-n}g^{-1} \in H$ and so
$\overline{ga^{m-n}g^{-1}} \in L({\cal{S}}(H))$ by Proposition~\ref{BS1:gwp}. Since $ga$ is reduced due to ${\cal{S}}(H)$ being
inverse, it follows from $m-n > |g|$ that
$ga\overline{a^{m-n-1}g^{-1}} = \overline{ga^{m-n}g^{-1}} \in
L({\cal{S}}(H))$: indeed, $g^{-1}$ is not long enough to erase all the
$a$'s. Since ${\cal{S}}(H)$ is deterministic, $q\cdot a$ must be
defined, a contradiction. Therefore~\eqref{BS1:findex1} holds and so
$H$ has infinite index.

$(\Leftarrow)$: Let $Q$ be the vertex set of ${\cal{S}}(H)$ and fix a
geodesic $q_0 \mapright{g_q} q$ for each $q \in Q$. Take $u \in
F_A$. Since ${\cal{S}}(H)$ is complete, we have a
path $q_0 \mapright{u}
q$ for some $q \in Q$. Hence $ug_q^{-1} \in H$ and so $u =
ug_q^{-1}g_q \in Hg_q$. Therefore $F_A = \bigcup_{q\in Q} Hg_q$ and so
$H \leq_{f.i.} F_A$.

(ii) In view of $F_A = \bigcup_{q\in Q} Hg_q$, it suffices to show that
the cosets $Hg_q$ are all distinct. Indeed, assume that $Hg_p = Hg_q$
for some vertices $p,q \in Q$. Then $g_pg_q^{-1} \in H$ and so
$\overline{g_pg_q^{-1}}
\in L({\cal{S}}(H))$ by Proposition~\ref{BS1:gwp}. On the other hand, since
${\cal{S}}(H)$ is complete, we have a path
\begin{displaymath}
q_0 \vlongmapright{g_pg_q^{-1}} r
\end{displaymath}
for some $r \in Q$. In view of
Proposition~\ref{BS1:invaut}, and by determinism, we get $r = q_0$. Hence
we have paths
\begin{displaymath}
p \longmapright{g_q^{-1}} q_0, \quad q \longmapright{g_q^{-1}}
q_0\, .
\end{displaymath}
Since ${\cal{S}}(H)$ is inverse, we get $p = q$ as required.
\end{proof}

\begin{example}
\label{BS1:exstal3}
Since the Stallings automaton constructed in Example~\ref{BS1:exstal}
is not complete, it follows that $\langle a^{-1}ba,
ba^2 \rangle$ is not a finite index subgroup of $F_2$.
\end{example}

\begin{corollary}
\label{BS1:rankfis}
If $H \leq F_A$ has index $n$, then $\rk H = 1+n(\Card{A} -1)$.
\end{corollary}

\begin{proof}
  By Proposition~\ref{BS1:findex}, the automaton ${\cal{S}}(H)$ has
  $n$ vertices and $n\Card{A}$ positive edges. A spanning tree has
  $n-1$ positive edges, so $\rk H = n\Card{A} -(n-1) = 1+n(\Card{A}
  -1)$ by Proposition~\ref{BS1:basis}.
\end{proof}

\index{Stallings'~construction|)}

Beautiful connections between finite index subgroups and certain
classes of {\em bifix codes}\index{bifix code}\index{code!bifix} ---
set of words none of which is a prefix or a suffix of another --- have
recently been unveiled by \Berstel, \DeFelice, \Perrin, \Reutenauer\
and \Rindone~\cite{Berstel&etal:2010}.

\subsection{Conjugacy}\label{BS1:conjugacy}

We start now a brief discussion of conjugacy.  Recall that the
\emph{outdegree} of a vertex $q$ is the number of edges starting at
$q$ and the \emph{geodesic
  distance}\index{geodesic~distance}\index{distance!geodesic} in a
connected graph is the length of the shortest undirected path
connecting two vertices.

Since the original generating
set is always taken in reduced
form, it follows easily that there is at most one vertex in a
Stallings automaton having outdegree $< 2$: the basepoint $q_0$.
Assuming
that $H$ is nontrivial, ${\cal{S}}(H)$ must always be of the form
$$\begin{fsa}
\node[state] (q_0) at (-2.3,0) {$q_0$};
\node[state] (q_1) at (0,0) {$q_1$};
\node[state] (NE) at (2,1.2) {};
\node[state] (SE) at (2,-1.2) {};
\node at (3,1.2) {$\cdots$};
\node at (3,0) {$\cdots$};
\node at (3,-1.2) {$\cdots$};
\path (q_0) edge node {$u$} (q_1);
\path (q_1) edge node {} (NE) edge node {} (SE);
\end{fsa}$$
where $q_1$ is the closest vertex to $q_0$ (in terms of geodesic
distance) having outdegree $> 2$ (since there is at least one
vertex having such outdegree). Note that $q_1 = q_0$ if $q_0$ has
outdegree $> 2$ itself. We call $q_0 \mapright{u}$ the \emph{tail}
(which is empty if $q_1 = q_0$) and the remaining subgraph the
\emph{core} of ${\cal{S}}(H)$.

Note that ${\cal S}(H)$, and its core, may be understood as
follows. Consider the graph with vertex set $F_A/H=\{gH\mid g\in
F_A\}$, with an edge from $gH$ to $agH$ for each generator $a\in
A$. Then this graph, called the \emph{Schreier
  graph}\index{Schreier~graph}\index{graph!Schreier}
(see~\S\ref{BS2:sec:1}) of $H\backslash F_A$, consists of finitely many trees
attached to the core of ${\cal S}(H)$.

\begin{theorem}
\label{BS1:conjsubgroups}
  There is an algorithm that decides whether or not two finitely
  generated subgroups of $F_A$ are conjugate.
\end{theorem}
\begin{proof}
  Finitely generated subgroups $G,H$ are conjugate if and only if the
  cores of $\mathcal S(G)$ and $\mathcal S(H)$ are equal (up to their
  basepoints).
\end{proof}

The Stallings automata of the conjugates of $H$ can be obtained in
the following ways: (1) declaring a vertex in the core $C$ to be the
basepoint; (2) gluing a tail to some vertex in the core $C$ and taking
its other endpoint to be the basepoint.

Note that the tail must be glued in some way that keeps the automaton
inverse, so in particular this second type of operation can only be
performed if the automaton is not complete, or equivalently, if $H$
has infinite index. An immediate consequence is the following
classical
\begin{proposition}
\label{normalsubgroups}
A finite rank normal subgroup\index{subgroup!normal} of a free group
is either trivial or has finite index.
\end{proposition}
Moreover, a finite index subgroup $H$ is normal if and only if its
Stallings automaton is \emph{vertex-transitive}, that is, if all
choices of basepoint yield the same automaton.

\begin{example}
\label{BS1:exstal4}
Stallings automata of some conjugates of $H = \langle a^{-1}ba,
ba^2 \rangle$:
\end{example}
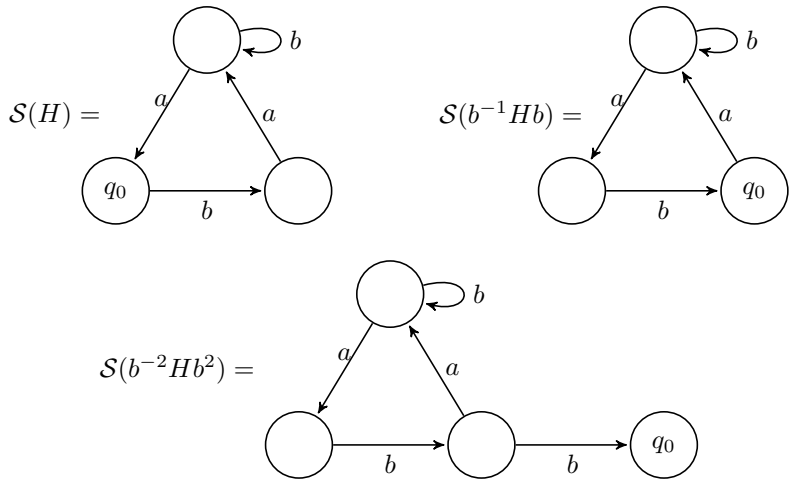
$$\begin{fsa}
\node[state] (p_0) at (-4.2,0) {$q_0$};
\node[state] (EE) at (-1.8,0) {};
\node[state] (NNEE) at (-3,2) {};
\node at (-5,1) {${\cal{S}}(H) =$};
\path (p_0) edge node[below] {$b$} (EE);
\path (NNEE) edge node[above] {$a$} (p_0) edge[loop right] node {$b$} ();
\path (EE) edge node[right] {$a$} (NNEE);
\node[state] (W) at (1.8,0) {};
\node[state] (q_0) at (4.2,0) {$q_0$};
\node[state] (N) at (3,2) {};
\node at (1,1) {${\cal{S}}(b^{-1}Hb) =$};
\path (W) edge node[below] {$b$} (q_0);
\path (N) edge node[above] {$a$} (W) edge[loop right] node {$b$} ();
\path (q_0) edge node[right] {$a$} (N);
\end{fsa}$$
$$\begin{fsa}
\node[state] (WW) at (-2.4,0) {};
\node[state] (W) at (0,0) {};
\node[state] (q_0) at (2.4,0) {$q_0$};
\node[state] (N) at (-1.2,2) {};
\node at (-4,1) {${\cal{S}}(b^{-2}Hb^2) =$};
\path (WW) edge node[below] {$b$} (W);
\path (N) edge node[above] {$a$} (WW) edge[loop right] node {$b$} ();
\path (W) edge node[right] {$a$} (N) edge node[below] {$b$} (q_0);
\end{fsa}$$
We can also use the previous discussion on the structure of (finite) Stallings
automata to provide them with an abstract characterization.

\begin{proposition}
\label{BS1:whoarethey}
A finite inverse automaton with a
basepoint is a Stallings
automaton if and only if it has at most one vertex of outdegree 1: the
basepoint.
\end{proposition}

\begin{proof}
Indeed, for any such automaton we can take a spanning tree
and use it to
construct a basis for the subgroup as in the proof of
Proposition~\ref{BS1:basis}.
\end{proof}

\subsection{Further algebraic properties}
\label{BS1:fap}

The study of intersections\index{subgroup!intersection} of finitely
generated subgroups of $F_A$ provides further applications of
Stallings automata. \Howson's classical theorem admits a simple proof
using the \emph{direct product} of two Stallings automata; it is also
an immediate consequence of Theorem \ref{BS1:anisei} and Corollary
\ref{BS1:bencon}(ii).

\begin{theorem}[Howson]
\label{BS1:howson}
If $H,K \leq_{f.g.} F_A$, then also $H \cap K \leq_{f.g.} F_A$.
\end{theorem}




Stallings automata are also naturally related to the famous Hanna
Neumann conjecture: given $H,K \leq_{f.g.} F_A$, then $\rk(H\cap K) -1
\leq (\rk H-1)(\rk K-1)$. The conjecture arose in a paper of
Hanna \Neumann~\cite{Neumann:1957}, where the inequality $\rk (H\cap
K) -1 \leq 2(\rk H-1)(\rk K-1)$ was also proved. In one of the early
applications of Stallings' approach, \Gersten\ provided an alternative
geometric proof of Hanna Neumann's inequality~\cite{MR695907}.

A \emph{free factor}\index{free~group!free~factor} of a free group
$F_A$ can be defined as a subgroup
$H$ generated by a subset of a basis of $F_A$. This is equivalent to
saying that there exists a \emph{free product decomposition} $F_A = H \ast
K$ for some $K \leq F_A$.

Since the rank of a free factor never exceeds the rank of the ambient
free group, it is easy to construct examples of subgroups which are
not free factors: it follows easily from Proposition~\ref{BS1:basis} that any
free group of rank $\geq 2$ can have subgroups
of arbitrary finite rank (and even infinite countable).

The problem of identifying free factors has a simple
solution based on Stallings automata \cite{Silva&Weil:2008}: one
must check whether or not a prescribed number of vertex
identifications in the Stallings automaton can lead to a bouquet.
However, the most efficient solution, due to \Roig, \Ventura\ and
\Weil~\cite{MR2378055}, involves \Whitehead\ automorphisms and will
therefore be postponed to~\S\chapterref{BS1:tdp}.

Given a morphism $\varphi:{\cal{A}} \to {\cal{B}}$ of inverse automata,
let the {\em morphic image} $\varphi({\cal{A}})$ be the subautomaton of
${\cal{B}}$ induced by the image by $\varphi$ of all the successful
paths of ${\cal{A}}$.

The following classical result characterizes the extensions of $H
\leq_{f.g.} F_A$ contained in $F_A$. We present the proof
from~\cite{MR2395796}:

\begin{theorem}[\Takahasi~\cite{Takahasi:1951}]
\label{BS1:takahasi}
Given $H \leq_{f.g.} F_A$, one can effectively compute finitely many extensions
$K_1, \dots, K_m \leq_{f.g.} F_A$ of $H$ such that the following
conditions are equivalent for every $K \leq_{f.g.} F_A$:
\begin{conditionsiii}
\item[\textup{(i)}] $H \leq K$;
\item[\textup{(ii)}] $K_i$ is a free factor of $K$ for some $i \in \{
  1,\dots, m\}$.
\end{conditionsiii}
\end{theorem}

\begin{proof}
Let ${\cal{A}}_1, \dots, {\cal{A}}_m$ denote all the morphic images
of ${\cal{S}}(H)$, up to isomorphism. Since a morphic image cannot
have more vertices than the original automaton, there are only
finitely many isomorphism classes. Moreover, it follows from Proposition~\ref{BS1:whoarethey}
that, for $i = 1,\dots, m$, ${\cal{A}}_i = {\cal{S}}(K_i)$ for some
$K_i \leq_{f.g.} F_A$. Since $L({\cal{S}}(H)) \subseteq L({\cal{A}}_i)
= L({\cal{S}}(K_i))$, it follows from Proposition~\ref{BS1:gwp} that $H \leq
K_i$. Clearly, we can construct all ${\cal{A}}_i$ and therefore all $K_i$.

(i) $\Rightarrow$ (ii).  If $H \leq K$, it follows from Stallings'
construction that $L({\cal{S}}(H))$ $\subseteq L({\cal{S}}(K))$ and so
there is a morphism $\varphi:{\cal{S}}(H) \to {\cal{S}}(K)$ by
Proposition~\ref{BS1:invmor}. Let ${\cal{A}}_i$ be, up to isomorphism,
the morphic image of ${\cal{S}}(H)$ through $\varphi$. Since
${\cal{A}}_i = {\cal{S}}(K_i)$ is a subautomaton of ${\cal{S}}(K)$, it
follows easily from Proposition~\ref{BS1:basis} that $K_i$ is a free
factor of $K$: it suffices to take a spanning tree for
${\cal{S}}(K_i)$, extend it to a spanning tree for ${\cal{S}}(K)$, and
the induced basis of $K_i$ will be contained in the induced basis of
$K$.

(ii) $\Rightarrow$ (i) is immediate.
\end{proof}

An interesting research line related to this result is built on the
concept of algebraic extension, introduced by
\Kapovich\ and \Miasnikov\ \cite{MR1882114}, and inspired by the homonymous
field-theoretical classical notion. Given $H \leq K \leq F_A$, we say
that $K$ is an \emph{algebraic}
extension\index{extension!algebraic}\index{algebraic~extension} of $H$
if no proper free factor of $K$ contains $H$. \Miasnikov, \Ventura\
and \Weil\ \cite{MR2395796} proved that the set of algebraic
extensions of $H$ is finite and effectively computable, and it
constitutes the minimum set of extensions $K_1, \dots, K_m$ satisfying
the conditions of Theorem~\ref{BS1:takahasi}.

Consider a subgroup $H$ of a group $G$. The
\emph{commensurator}\index{commensurator~of~group} of $H$ in $G$, is
\begin{equation}\label{BS1:eq:commensurator}
  \operatorname{Comm}_G(H)=\{g\in G\mid H\cap H^g\text{ has finite index in }H\text{
  and }H^g\}.
\end{equation}
For example, the commensurator of $\GL_n(\Z)$ in $\GL_n(\R)$ is
$\GL_n(\Q)$.

The special case of finite-index
extensions\index{extension!finite-index}, $H \leq_{f.i.} K \leq F_A$
is of special interest, and can be interpreted in terms of
commensurators. It can be proved (see~\cite[Lemma 8.7]{MR1882114}
and~\cite{Silva&Weil:2008a}) that every $H \leq_{f.g.} F_A$ has a
maximum finite-index extension inside $F_A$, denoted by $H_{fi}$; and
$H_{fi}=\operatorname{Comm}_{F_A}(H)$.  \Silva\ and \Weil\
\cite{Silva&Weil:2008a} proved
 that ${\cal{S}}(H_{fi})$ can be constructed
from ${\cal{S}}(H)$ using a simple automata-theoretic algorithm:
\begin{enumerate}
\item
The standard minimization algorithm is applied to the core of
${\cal{S}}(H)$, \emph{taking all vertices as final}.
\item
The original tail of ${\cal{S}}(H)$ is subsequently reinstated in this new
automaton, at the appropriate vertex.
\end{enumerate}

We present now an application of different type, involving transition
monoids. It follows easily from the definitions that the transition
monoid\index{transition~monoid}\index{monoid!transition} of a finite
inverse automaton is always a \emph{finite inverse
  monoid}\index{monoid!inverse}.  Given a group $G$, we say that a
subgroup $H \leq G$ is
\emph{pure}\index{subgroup!pure@[$p$-]pure}\index{pure@[$p$-]pure~subgroup}
if the implication
\begin{equation}
\label{BS1:pure}
g^n \in H \Rightarrow g \in H
\end{equation}
holds for all $g \in F_A$ and $n \geq 1$. If $p$ is a prime, we say
that $H$ is $p$-\emph{pure} if~\eqref{BS1:pure} holds when $(n,p)=1$.

The next result is due to \Birget, \Margolis, \Meakin\ and \Weil, and
is the only natural problem among applications of Stallings automata
that is known so far to be PSPACE-complete~\cite{MR1769781}.

\begin{proposition}
\label{BS1:apsynt}
For every $H \leq_{f.g.} F_A$, the following conditions hold:
\begin{conditionsiii}
\item[\textup{(i)}] $H$ is pure if and only if the transition monoid of
  ${\cal{S}}(H)$ is aperiodic\index{monoid!aperiodic}\index{aperiodic~monoid}.
\item[\textup{(ii)}] $H$ is $p$-pure if and only if the transition monoid of
  ${\cal{S}}(H)$ has no subgroups of order $p$.
\end{conditionsiii}
\end{proposition}

\begin{proof}
Both conditions in (i) are easily proved to be equivalent to the
nonexistence in ${\cal{S}}(H)$ of a cycle of the form
$$\begin{fsa}
\node[state] (p) at (-3,0) {$p$};
\node[state] (q) at (-1,0) {$q$};
\node at (2,0) {$(k \geq 1, p \neq q)$};
\path (p) edge[bend left] node {$u$} (q);
\path (q) edge[bend left] node[below] {$u^k$} (p);
\end{fsa}$$
where $u$ can be assumed to be cyclically reduced.
The proof of (ii) runs similarly.
\end{proof}

\subsection{Topological properties}

We require for this subsection some basic topological concepts, which the
reader can recover from Chapter~17.

For all $u,v \in F_A$, written in reduced form as elements of $R_A$,
let $u \wedge v$ denote the longest common prefix of $u$ and $v$. The
\emph{prefix metric}\index{prefix~metric}\index{distance!prefix~metric} $d$ on $F_A$ is defined, for
all $u,v \in F_A$,
by
\begin{displaymath}
d(u,v) = \left\{
\begin{array}{ll}
2^{-|u \wedge v|-1}&\mbox{ if }u \neq v\\
0&\mbox{ if }u = v
\end{array}
\right.
\end{displaymath}
It follows easily from the definition
that $d$ is an
ultrametric on $F_A$, satisfying in particular the axiom
\begin{displaymath}
d(u,v) \leq \max\{ d(u,w), d(w,v)
\}\, .
\end{displaymath}
The \emph{completion} of this metric space is compact; its extra
elements are \emph{infinite reduced words} $a_1a_2a_3\dots$, with all
$a_i \in \widetilde{A}$, and constitute the \emph{hyperbolic
  boundary}\index{hyperbolic!boundary}\index{boundary!hyperbolic}
$\partial F_A$ of $F_A$, see~\S\ref{BS2:sec:hypgroups}. Extending the
operator $\wedge$ to $F_A\cup\partial F_A$ in the obvious way, it
follows easily from the definitions that, for every infinite reduced
word $\alpha$ and every sequence $(u_n)_n$ in $F_A$,
\begin{equation}
\label{BS1:limits}
\alpha = \displaystyle\lim_{n\to +\infty} u_n \hspace{1cm}\mbox{ if and only if }
\hspace{1cm} \displaystyle\lim_{n\to +\infty} |\alpha \wedge u_n| =
+\infty\, .
\end{equation}

The next result shows that Stallings automata are given a new role in
connection with the prefix metric. We denote by $\cl H$ the closure of
$H$ in the completion of $F_A$.

\begin{proposition}
\label{BS1:clogeo}
If $H \leq_{f.g.} F_A$, then $\cl H$ is the union of $H$ with the set of all
$\alpha \in \partial F_A$ that label paths in ${\cal{S}}(H)$ out of the
basepoint.
\end{proposition}

\begin{proof}
Since the topology of $F_A$ is discrete, we have
$\cl H \cap F_A = H$.

$(\subseteq)$: If $\alpha
\in \partial F_A$ does not label a path in ${\cal{S}}(H)$ out of the
basepoint, then $\{ |\alpha \wedge h|:\; h \in H\}$ is finite and so
no sequence of $H$ can converge to $\alpha$ by~\eqref{BS1:limits}.

$(\supseteq)$: Let $\alpha = a_1a_2a_3\dots
\in \partial F_A$, with $a_i \in \widetilde{A}$,
label a path in ${\cal{S}}(H)$ out of the
basepoint. Let $m$ be the number of vertices of ${\cal{S}}(H)$. For
every $n \geq 1$, there exists some word $w_n$ of length $< m$ such
that $a_1\cdots a_nw_n \in H$. Now $\alpha = \lim_{n\to +\infty}
a_1\cdots a_nw_n$ by~\eqref{BS1:limits} and so $\alpha \in \cl H$.
\end{proof}

The \emph{profinite topology} on $F_A$ is defined in Chapter~17: for
every $u \in F_A$, the collection $\{ Ku \mid K \leq_{f.i.} F_A\}$
constitutes a basis of clopen neighbourhoods of $u$.  In his seminal
1983 paper~\cite{MR695906}, Stallings gave an alternative proof of
Marshall \Hall's Theorem:

\begin{theorem}[M.\ Hall]
\label{BS1:halls}
Every finitely generated subgroup of $F_A$ is closed for the profinite
topology.
\end{theorem}

\begin{proof}
  Fix $H \leq_{f.g.} F_A$ and let $u \in F_A \setminus H$ be written
  in reduced form as an element of $R_A$. In view of Proposition
 ~\ref{BS1:gwp}, $u$ does not label a loop at the basepoint $q_0$ of
  ${\cal{S}}(H)$. If there is no path $q_0 \mapright{u} \cdots$ in
  ${\cal{S}}(H)$, we add new edges to ${\cal{S}}(H)$ to get a finite
  inverse automaton ${\cal{A}}$ having a path $q_0 \mapright{u} q \neq
  q_0$. Otherwise just take ${\cal{A}} = {\cal{S}}(H)$. Next add new
  edges to ${\cal{A}}$ to get a finite complete inverse automaton
  ${\cal{B}}$. In view of Propositions~\ref{BS1:findex} and
 ~\ref{BS1:whoarethey}, we have ${\cal{B}} = {\cal{S}}(K)$ for some $K
  \leq_{f.i.} F_A$. Hence $Ku$ is open and contains $u$. Since $H \cap
  Ku \neq \emptyset$ yields $u \in K^{-1}H = K$, contradicting
  Proposition~\ref{BS1:gwp}, it follows that $H \cap Ku = \emptyset$
  and so $H$ is closed as claimed.
\end{proof}

\begin{example}
\label{BS1:exhall}
We consider the above construction for $H = \langle a^{-1}ba, ba^2
\rangle$ and $u = b^2$:
$$\begin{fsa}
\node[state] (p_0) at (-4.2,0) {$q_0$};
\node[state] (EE) at (-1.8,0) {};
\node[state] (NNEE) at (-3,2) {};
\node at (-5,1) {${\cal{S}}(H) =$};
\path (p_0) edge node[below] {$b$} (EE);
\path (NNEE) edge node[above] {$a$} (p_0) edge[loop right] node {$b$} ();
\path (EE) edge node[right] {$a$} (NNEE);
\node[state] (W) at (1.8,0) {$q_0$};
\node[state] (q_0) at (4.2,0) {};
\node[state] (N) at (3,2) {};
\node[state] (x) at (6,0) {};
\node at (1,1) {${\cal{A}} =$};
\path (W) edge node[below] {$b$} (q_0);
\path (N) edge node[above] {$a$} (W) edge[loop right] node {$b$} ();
\path (q_0) edge node[right] {$a$} (N) edge node[below] {$b$} (x);
\end{fsa}
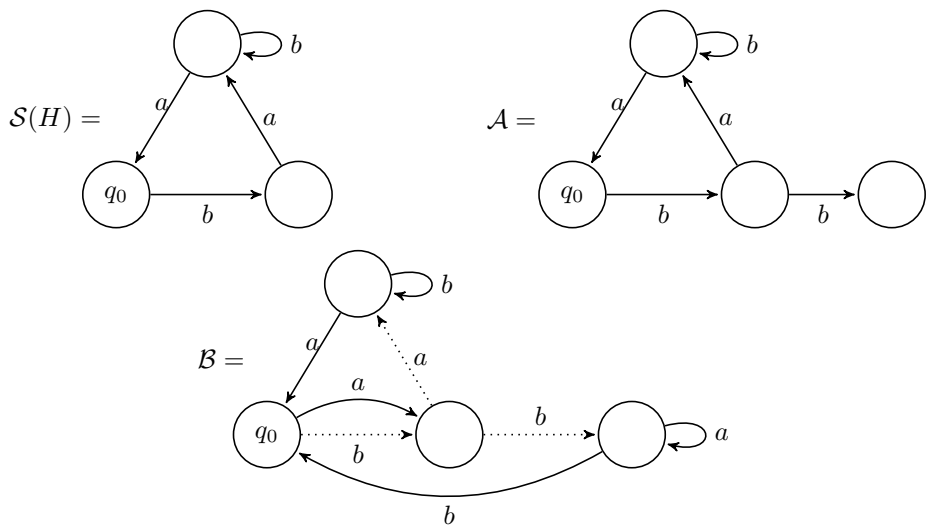
$$
$$\begin{fsa}
\node[state] (WW) at (-2.4,0) {$q_0$};
\node[state] (W) at (0,0) {};
\node[state] (q_0) at (2.4,0) {};
\node[state] (N) at (-1.2,2) {};
\node at (-3,1) {${\cal{B}} =$};
\path (WW) edge[bend left] node {$a$} (W) edge[dotted] node[below] {$b$} (W);
\path (N) edge node[above] {$a$} (WW) edge[loop right] node {$b$} ();
\path (W) edge[dotted] node {$b$} (q_0) edge[dotted] node[right] {$a$} (N);
\path (q_0) edge[bend left] node[below] {$b$} (WW) edge[loop right]
node {$a$} ();
\end{fsa}$$
If we take the spanning tree defined by the dotted lines in
${\cal{B}}$, it follows from Proposition~\ref{BS1:basis} that
\begin{displaymath}
K = \langle
ba^{-1}, b^3, b^2ab^{-2}, ba^2, baba^{-1}b^{-1}
\rangle
\end{displaymath}
is a finite index subgroup of $F_2$ such that $H \cap Kb^2 =
\emptyset$.
\end{example}

We recall that a group $G$ is \emph{residually
  finite}\index{group!residually~finite}\index{residually~finite~group}
if its finite index subgroups have trivial intersection. Considering
the trivial subgroup in Theorem~\ref{BS1:halls}, we deduce
\begin{corollary}
\label{BS1:rfin}
$F_A$ is residually finite\index{free~group!is~residually~finite}.
\end{corollary}

We remark that \Ribes\ and \Zalessky\ extended Theorem~\ref{BS1:halls}
to products of finitely many finitely generated subgroups of $F_A$,
see~\cite{Ribes&Zalesskii:1993}. This result is deeply connected to
the solution of \Rhodes'\ Type II conjecture, see~\cite[Chapter
4]{Rhodes&Steinberg:2009}.

If {\bf V} denotes a pseudovariety of finite groups (see Chapter~16),
the \emph{pro-}{\bf V}
\emph{topology}\index{topology!pro-V@pro-\textbf{V}} on $F_A$ is defined
by considering that each $u \in F_A$ has
\begin{displaymath}
\{ Ku \mid K \unlhd_{f.i.} F_A,\; F_A/K \in {\bf V} \}
\end{displaymath}
as a basis of clopen neighbourhoods.  The closure for the pro-{\bf V}
topology of $H \leq_{f.g} F_A$ can be related to an extension property
of ${\cal{S}}(H)$, and \Margolis, \Sapir\ and \Weil\ used automata to
prove that efficient computation can be achieved for the
pseudovarieties of finite
$p$-groups\index{group!p-@$p$-}\index{p-group@$p$-group} and finite
nilpotent\index{group!nilpotent}\index{nilpotent~group}
groups~\cite{Margolis&Sapir&Weil:2001}. The original computability
proof for the $p$-group case is due to \Ribes\ and
\Zalessky~\cite{Ribes&Zalesskii:1994}.

\subsection{Dynamical properties}
\label{BS1:tdp}

We shall mention briefly some examples of applications of Stallings
automata to the study of endomorphism dynamics, starting with
\Gersten's solution of the subgroup orbit problem~\cite{MR733696}.

The subgroup orbit problem consists in finding an algorithm to decide,
for given $H,K$ $\leq_{f.g.} F_A$, whether or not $K = \varphi(H)$ for
some automorphism $\varphi$ of $F_A$. Equivalently, this can be
described as deciding whether or not the automorphic orbit of a
finitely generated subgroup is recursive.

Gersten's solution adapts to the context of Stallings automata
\Whitehead's\ idea to solve the orbit problem for
words~\cite{Whitehead:1936}.  Whitehead's proof relies on a suitable
decomposition of automorphisms as products of elementary factors
(which became known as \emph{Whitehead
  automorphisms}\index{automorphism!Whitehead}), and on using these as
a tool to compute the elements of minimum length in the automorphic
orbit of the word. In the subgroup case, word length is replaced by
the number of vertices of the Stallings automaton.

The most efficient solution to the problem of identifying free
factors~\cite{MR2378055}, mentioned
in~\S\chapterref{BS1:fap}, also relies on this approach: $H
\leq_{f.g.} F_A$ is a free factor if and only if the Stallings
automaton of some automorphic image of $H$ has a single vertex (that
is, a bouquet).

Another very nice application is given by the following theorem of
\Goldstein\ and \Turner~\cite{MR847985}:
\begin{theorem}
\label{BS1:fixed}
The fixed point subgroup\index{subgroup!fixed~point} of an
endomorphism of $F_A$ is finitely
generated.
\end{theorem}

\begin{proof}
Let $\varphi$ be an endomorphism of $F_A$. For every $u \in F_A$, define
$Q(u) = \varphi(u)u^{-1}$. We define a potentially infinite
automaton $\cal{A}$ by taking
\begin{displaymath}
\{ Q(u) \mid u \in F_A\} \subseteq F_A
\end{displaymath}
as the vertex set, all edges of the form $Q(u) \mapright{a}
Q(au)$ with $u\in F_A, \; a \in \widetilde{A}$, and fixing
$\one$ as the basepoint. Then $\cal{A}$ is a well-defined inverse automaton.

Next we take $\cal{B}$ to be the subautomaton of $\cal{A}$ obtained
by retaining only those vertices and edges that lie in successful
paths labelled by reduced words. Clearly, $\cal{B}$ is still an
inverse automaton, and it is easy to check that it must be the
Stallings automaton of the fixed point subgroup of $\varphi$.

It remains to be proved that $\cal{B}$ is finite. We define a
subautomaton $\cal{C}$ of $\cal{B}$ by removing exactly one edge among
each inverse pair
\begin{displaymath}
Q(u) \mapright{a} Q(au),\quad Q(au) \longmapright{a^{-1}} Q(u)
\end{displaymath}
with $a\in A$ as follows:
if $a^{-1}$ is the last letter of $Q(au)$, we remove $Q(u)
\mapright{a} Q(au)$;
otherwise, we remove $Q(au) \longmapright{a^{-1}}
Q(u)$.

Let $M$ denote the
maximum length of the image of a letter by $\varphi$. We claim that,
whenever $|Q(v)| > 2M$, the vertex $Q(v)$ has outdegree at most 1.

Indeed, if
$Q(v) \longmapright{a^{-1}} Q(a^{-1}v)$ is an edge in $\cal{C}$ for $a \in
A$, then $a^{-1}$ is the last letter of $Q(v)$. On the other hand, if
$Q(v) \mapright{b} Q(bv)$ is an edge in $\cal{C}$ for $b \in
A$, then $b^{-1}$ is not the last letter of $Q(bv)$. Since $Q(bv) =
\varphi(b)Q(v)b^{-1}$ and $|Q(v)| > 2|\varphi(b)|$, then $b$ must be
the last letter of $Q(v)$ in this case. Since $Q(v)$ has at most one
last letter, it follows that its outdegree is at most 1.

Let $\cal{D}$ be a finite subautomaton of $\cal{C}$ containing all
vertices $Q(v)$ such that $|Q(v)| \leq 2M$. Suppose that $p
\mapright{} q$ is an edge in $\cal{C}$ not belonging to
$\cal{D}$. Since $p
\mapright{} q$, being an edge of $\cal{B}$, must lie in some reduced
path, and by the outdegree property of $\cal{C}$, it is easy to see
that there exists some path in $\cal{C}$ of the form
\begin{displaymath}
p' \mapright{} p \mapright{} q \mapright{} r \mapleft{} r'
\end{displaymath}
where $p',r'$ are vertices in $\cal{D}$. Since there are only finitely
many directed paths out of $\cal{D}$, it follows that $\cal{C}$ is
finite and so is $\cal{B}$. Therefore the fixed point subgroup of
$\varphi$ is finitely generated.
\end{proof}





Note that this proof is not by any means constructive. Indeed, the
only known algorithm for computing the fixed point subgroup of a free
group automorphism is due to \Maslakova~\cite{Maslakova:2003}\ and
relies on the sophisticated \emph{train track} theory of \Bestvina\
and \Handel~\cite{Bestvina&Handel:1992} and other algebraic geometry
tools. The general endomorphism case remains open.

Stallings automata were also used by \Ventura\ in the study of various
properties of fixed subgroups, considering in particular arbitrary
families of endomorphisms~\cite{Ventura:1997,Martino&Ventura:2004}
(see also~\cite{Ventura:2002}).

Automata also play a part in the study of \emph{infinite fixed
  points}, taken over the continuous extension of a
monomorphism\index{monomorphism~extension} to the hyperbolic boundary
(see for example~\cite{Silva:2009}).

\section{Rational and recognizable subsets}\label{BS1:rational}

Rational subsets generalize the notion of finitely generated from
subgroups to arbitrary subsets of a group, and can be quite useful in
establishing inductive procedures that need to go beyond the territory
of subgroups. Similarly, recognizable subsets extend the notion of
finite index subgroups. Basic properties and results can be found
in~\cite{MR549481} or~\cite{Sakarovitch:2003}.

We consider a finitely generated group $G=\langle A\rangle$, with the
canonical map $\pi:F_A\to G$. A subset of $G$ is
\emph{rational}\index{rational~language}\index{language!rational} if
it is the image by $\rho=\pi\theta$ of a rational subset of
$\widetilde A^*$, and is
\emph{recognizable}\index{recognizable~language}\index{language!recognizable}
if its full preimage under $\rho$ is rational in $\widetilde A^*$.

For every group $G$, the classes $\Rat G$ and $\Rec G$ satisfy the
following closure properties:
\begin{itemize}
\item
$\Rat G$ is (effectively) closed under union, product, star, morphisms,
inversion, subgroup generating.
\item
$\Rec G$ is (effectively) closed under boolean operations, translation,
product, star, inverse morphisms, inversion, subgroup generating.
\end{itemize}
\Kleene's Theorem is not valid for groups: $\Rat G = \Rec G$ if and
only if $G$ is finite. However, if the class of rational subsets of
$G$ possesses some extra algorithmic properties, then many
decidability/constructibility results can be deduced for $G$. Two
properties are particularly coveted for $\Rat G$:
\begin{itemize}
\item
(effective) closure under complement (yielding closure under all the boolean
operations);
\item
decidable membership problem for arbitrary rational subsets.
\end{itemize}
In these cases, one may often solve problems (e.g.\ equations, or
systems of equations) whose statement lies far out of the rational
universe, by proving that the solution is a rational set.

\subsection{Rational and recognizable subgroups}

We start by some basic, general facts.  The following result is
essential to connect language theory to group theory.

\begin{theorem}[\Anisimov\ and \Seifert]
\label{BS1:anisei}
A subgroup $H$ of a group $G$ is rational if and only if $H$ is finitely
generated.
\end{theorem}

\begin{proof}
($\Rightarrow$): Let $H$ be a rational subgroup of $G$ and let $\pi:
F_A \to G$ denote a morphism. Then there
exists a finite $\widetilde{A}$-automaton ${\cal{A}}$ such that $H =
\rho(L({\cal{A}}))$. Assume that ${\cal{A}}$ has $m$ vertices and
let $X$ consist of all the words in $\rho^{-1}(H)$ of length $<
2m$. Since $A$ is finite, so is $X$. We claim that $H = \langle
\rho(X)\rangle$. To prove it, it suffices to show that
\begin{equation}
\label{BS1:anisei1}
u \in L({\cal{A}}) \Rightarrow \rho(u) \in \langle
\rho(X)\rangle
\end{equation}
holds for every $u \in \widetilde{A}^*$. We use induction on $|u|$. By
definition of $X$,~\eqref{BS1:anisei1} holds for words of length $<
2m$. Assume now that $|u| \geq 2m$ and~\eqref{BS1:anisei1} holds for
shorter words. Write $u = vw$ with $|w| = m$. Then there exists a path
\begin{displaymath}
\to q_0 \mapright{v} q \mapright{z} t \to
\end{displaymath}
in ${\cal{A}}$ with $|z| < m$. Thus $vz \in L({\cal{A}})$ and by the
induction hypothesis $\rho(vz) \in \langle
\rho(X)\rangle$. On the other hand, $|z^{-1}w| < 2m$ and
$\rho(z^{-1}w) = \rho(z^{-1}v^{-1})\rho(vw) \in H$, hence
$z^{-1}w \in X$ and so $\rho(u) =
\rho(vz)\rho(z^{-1}w)\in \langle
\rho(X)\rangle$, proving~\eqref{BS1:anisei1} as required.

($\Leftarrow$) is trivial.
\end{proof}

\noindent It is an easier task to characterize recognizable subgroups:

\begin{proposition}
\label{BS1:recsub}
A subgroup $H$ of a group $G$ is recognizable if and only if it has
finite index.
\end{proposition}

\begin{proof}
($\Rightarrow$):
In general, a recognizable subset of $G$ is of the form $NX$, where $N
\unlhd_{f.i.} G$ and $X \subseteq G$ is finite. If $H = NX$ is a
subgroup of $G$, then $N \subseteq H$ and so $H$ has finite index as
well.

($\Leftarrow$): This follows from the well-known fact that every
finite index subgroup $H$ of $G$ contains a finite index normal
subgroup $N$ of $G$, namely $N = \bigcap_{g\in G} \, gHg^{-1}$. Since $N$ has
finite index, $H$ must be of the form $NX$ for some finite $X
\subseteq G$.
\end{proof}

\subsection{Benois' Theorem}

The central result in this subsection is Benois' Theorem, the
cornerstone of the whole theory of rational subsets of free groups:

\begin{theorem}[\Benois]
\label{BS1:benois}
\begin{conditionsiii}
\item[\textup{(i)}] If $L \subseteq \widetilde{A}^*$ is
rational, then $\overline{L}$ is also rational, and can be effectively
constructed from $L$.
\item[\textup{(ii)}]
A subset of $R_A$ is a rational language as a subset of
$\widetilde{A}^*$ if and only if it is rational as a subset of $F_A$.
\end{conditionsiii}
\end{theorem}

We illustrate this in the case of finitely generated subgroups:
temporarily calling ``Benois automata'' those automata recognizing
rational subsets of $R_A$, we may convert them to Stallings automata
by ``folding'' them, at the same time making sure they are inverse
automata. Given a Stallings automaton, one intersects it with $R_A$ to
obtain a Benois automaton.

\begin{proof}
  (i) Let ${\cal{A}} = (Q,\widetilde{A},E,I,T)$ be a finite automaton
  recognizing 
  $L$. We define a sequence $({\cal{A}}_n)_n$ of finite automata with
  $\varepsilon$-transitions as follows. Let ${\cal{A}}_0 =
  {\cal{A}}$. Assuming that ${\cal{A}}_n = (Q,\widetilde{A},E_n,I,T)$
  is defined, we 
  consider all instances of ordered pairs $(p,q) \in Q \times Q$ such
  that
\[\tag{P}
  \text{there exists a path $p \longmapright{aa^{-1}} q$ in ${\cal{A}}_n$
    for some $a \in \widetilde{A}$, but no path $p \mapright{1} q$.}
\]
Clearly, there are only finitely many instances of (P) in
${\cal{A}}_n$. We define $E_{n+1}$ to be the union of $E_n$ with all the
new edges $(p,1,q)$, where $(p,q) \in Q \times Q$ is an instance of (P).
Finally, we define ${\cal{A}}_{n+1} =
(Q,\widetilde{A},E_{n+1},I,T)$. In particular, note that ${\cal{A}}_{n} =
{\cal{A}}_{n+k}$ for every $k \geq 1$ if there are no
instances of (P) in ${\cal{A}}_{n}$.

Since $Q$ is finite, the sequence $({\cal{A}}_n)_n$ is ultimately
constant, say after reaching ${\cal{A}}_m$. We claim that
\begin{equation}
\label{BS1:benois1}
\overline{L} = L({\cal{A}}_m) \cap R_A\, .
\end{equation}
Indeed, take $u \in L$. There exists a sequence of words $u = u_0, u_1,
\dots, u_{k-1}, u_k = \overline{u}$ where each term is obtained from the
preceding one by erasing a factor of the form $aa^{-1}$ for some $a
\in \widetilde{A}$. A straightforward induction shows that $u_i \in
L({\cal{A}}_i)$ for $i = 0, \dots,k$, since the existence of a path $p
\longmapright{aa^{-1}} q$ in ${\cal{A}}_i$ implies the existence of a
  path $p \mapright{1} q$ in ${\cal{A}}_{i+1}$. Hence $\overline{u} = u_k
  \in L({\cal{A}}_k) \subseteq L({\cal{A}}_m)$ and it follows that
  $\overline{L} \subseteq L({\cal{A}}_m) \cap R_A$.

For the opposite inclusion, we start by noting that any
path $p \mapright{u} q$ in ${\cal{A}}_{i+1}$ can be lifted to a path
$p \mapright{v} q$ in ${\cal{A}}_{i}$, where $v$ is obtained from $u$
by inserting finitely many factors of the form $aa^{-1}$. It follows that
\begin{displaymath}
\overline{L({\cal{A}}_m)} =
\overline{L({\cal{A}}_{m-1})} = \dots = \overline{L({\cal{A}}_0)} =
\overline{L}
\end{displaymath}
and so $L({\cal{A}}_m) \cap R_A \subseteq \overline{L({\cal{A}}_m)} =
\overline{L}$. Thus~\eqref{BS1:benois1} holds.

\noindent Since
\begin{displaymath}
R_A = \widetilde{A}^* \setminus \bigcup_{a \in \widetilde{A}}
\widetilde{A}^*aa^{-1}\widetilde{A}^*
\end{displaymath}
is obviously rational, and the class of rational languages is closed under
intersection,
it follows that $\overline{L}$ is rational. Moreover, we can
effectively compute the automaton ${\cal{A}}_m$ and a finite automaton
recognizing $R_A$, hence the direct product
construction can be used to construct a finite automaton recognizing
the intersection
$\overline{L}  = L({\cal{A}}_m) \cap R_A$.

(ii) Consider $X \subseteq R_A$. If $X \in
\Rat\widetilde{A}^*$, then $\theta(X) \in \Rat F_A$ and so
 $X$ is rational as a subset of $F_A$.

Conversely, if $X$ is rational as a subset of $F_A$, then $X =
\theta(L)$ for some $L \in
\Rat\widetilde{A}^*$. Since $X \subseteq R_A$, we get $X =
\overline{L}$.
Now part (i) yields $\overline{L} \in
\Rat\widetilde{A}^*$ and so $X \in
\Rat\widetilde{A}^*$ as required.
\end{proof}

\begin{example}
\label{BS1:exbenois}
Let ${\cal{A}} = {\cal{A}}_0$ be depicted by
$$\begin{fsa}
\node[state,initial] (1) at (0,0) {};
\node[state] (2) at (0,2) {};
\node[state] (3) at (2,2) {};
\node[state] (4) at (2,0) {};
\path (1) edge (180:1);
\path (1) edge node[right] {$a$} (2);
\path (2) edge node[above] {$a$} (3) edge[loop left] node {$b$} ();
\path (3) edge[loop right] node {$b$} () edge node[right] {$a^{-1}$} (4);
\path (4) edge node[below] {$b^{-1}$} (1);
\end{fsa}$$
We get
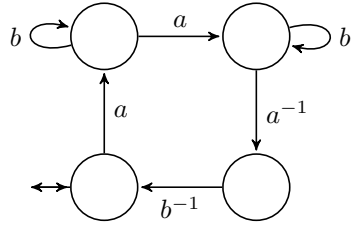
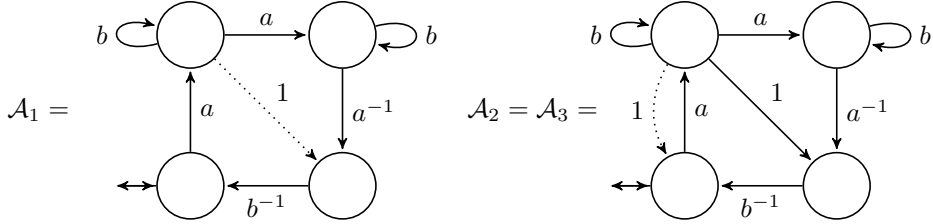
$$\begin{fsa}
\node[state,initial] (1) at (-5,0) {};
\node[state] (2) at (-5,2) {};
\node[state] (3) at (-3,2) {};
\node[state] (4) at (-3,0) {};
\node at (-7,1) {${\cal{A}}_1 =$};
\path (1) edge (0:-6);
\path (1) edge node[right] {$a$} (2);
\path (2) edge node[above] {$a$} (3) edge[dotted] node {$1$} (4)
 edge[loop left] node {$b$} ();
\path (3) edge[loop right] node {$b$} () edge node[right] {$a^{-1}$} (4);
\path (4) edge node[below] {$b^{-1}$} (1);
\node[state,initial] (5) at (1.5,0) {};
\node[state] (6) at (1.5,2) {};
\node[state] (7) at (3.5,2) {};
\node[state] (8) at (3.5,0) {};
\node at (-0.5,1) {${\cal{A}}_2 = {\cal{A}}_3 =$};
\path (5) edge (0:0.5);
\path (5) edge node[right] {$a$} (6);
\path (6) edge node[above] {$a$} (7) edge node {$1$} (8)
edge[bend right,dotted] node[left] {$1$} (5) edge[loop left] node {$b$} ();
\path (7) edge[loop right] node {$b$} () edge node[right] {$a^{-1}$} (8);
\path (8) edge node[below] {$b^{-1}$} (5);
\end{fsa}$$
and we can then proceed to compute $\overline{L} = L({\cal{A}}_2) \cap R_2$.
\end{example}

The following result summarizes some of the most direct consequences
of Benois' Theorem:

\begin{corollary}
\label{BS1:bencon}
\begin{conditionsiii}
\item[\textup{(i)}] $F_A$ has decidable rational
subset membership problem.
\item[\textup{(ii)}] $\Rat F_A$ is closed under
the boolean operations.
\end{conditionsiii}
\end{corollary}

\begin{proof}
  (i) Given $X \in \Rat F_A$ and $u \in F_A$, write $X = \theta(L)$
  for some $L \in \Rat\widetilde{A}^*$. Then $u \in X$ if and only if
  $\overline{u} \in \overline{X} = \overline{L}$. By
  Theorem~\ref{BS1:benois}(i), we may construct a finite automaton
  recognizing $\overline{L}$ and therefore decide whether or not
  $\overline{u} \in \overline{L}$.

  (ii) Given $X \in \Rat F_A$, we have $\overline{F_A \setminus X} =
  R_A \setminus \overline{X}$ and so $F_A \setminus X \in \Rat F_A$ by
  Theorem~\ref{BS1:benois}. Therefore $\Rat F_A$ is closed under
  complement.

Since $\Rat F_A$ is trivially closed under union, it follows from De
Morgan's laws that it is closed under intersection as well.
\end{proof}

Note that we can associate algorithms to these boolean closure
properties of $\Rat F_A$ in a constructive way.  We remark also that
the proof of Theorem~\ref{BS1:benois} can be clearly adapted to more
general classes of rewriting systems
(see~\cite{Book&Otto:1993}). Theorem~\ref{BS1:benois} and
Corollary~\ref{BS1:bencon} have been generalized several times by
\Benois\ herself~\cite{MR903667} and by \Senizergues, who obtained the
most general versions. \Senizergues' results~\cite{MR1057769} hold for
\emph{rational length-reducing left basic confluent} rewriting
systems\index{rewriting~system!confluent!length-reducing~etc.} and
remain valid for the more general notion of \emph{controlled}
rewriting system.

\subsection{Rational versus recognizable}

Since $F_A$ is a finitely generated monoid, it follows that every
recognizable subset of $F_A$ is rational~\cite[Proposition
III.2.4]{MR549481}.  We turn to the problem of deciding
which rational subsets of $F_A$ are recognizable. The first proof,
using rewriting systems, is due to \Senizergues~\cite{MR1393764} but
we follow the shorter alternative proof from~\cite{MR2059028}, where a
third alternative proof, of a more combinatorial nature, was also
given.

Given a subset $X$ of a group $G$, we define the \emph{right
  stabilizer} of
$X$ to be the submonoid of $G$ defined by
\begin{displaymath}R(X) = \{ g \in G \mid Xg \subseteq X\}\, .
\end{displaymath}
Next let $$K(X) = R(X) \cap (R(X))^{-1} = \{ g \in G \mid Xg = X\}$$
be the largest subgroup of $G$ contained in $R(X)$ and let
\begin{displaymath}
N(X) = \bigcap_{g \in G} gK(X)g^{-1}
\end{displaymath}
be the largest normal subgroup of $G$ contained in $K(X)$, and
therefore in $R(X)$.

\begin{lemma}
\label{BS1:rsn}
A subset $X$ of a group $G$ is recognizable if and only if $K(X)$ is a
finite index subgroup of $G$.
\end{lemma}

In fact, the Schreier
graph\index{Schreier~graph}\index{graph!Schreier}
(see~\S\ref{BS2:sec:1}) of $K(X)\backslash G$ is the underlying graph of an automaton recognizing $X$,
and $G/N(X)$ is the syntactic
monoid\index{syntactic~monoid}\index{monoid!syntactic} of $X$.

\begin{proof}
($\Rightarrow$): If $X \subseteq G$ is recognizable, then $X = NF$ for
some $N \unlhd_{f.i.} G$ and $F \subseteq G$ finite. Hence $N
\subseteq R(X)$ and so $N
\subseteq K(X)$ since $N \leq G$. Since $N$ has finite index in $G$, so does
$K(X)$.

($\Leftarrow$): If $K(X)$ is a
finite index subgroup of $G$, so is $N = N(X)$. Indeed, a finite index
subgroup has only finitely many conjugates (having also finite index)
and a finite intersection
of finite index subgroups is easily checked to have finite
index itself.

Therefore it suffices to show that $X = FN$ for some finite subset $F$
of $G$. Since $N$ has finite index, the claim follows from $XN=X$, in
turn an immediate consequence of $N \subseteq R(X)$.
\end{proof}

\begin{proposition}
\label{BS1:decrec}
It is decidable whether or not a rational subset of $F_A$ is
recognizable.
\end{proposition}

\begin{proof}
Take $X \in \Rat F_A$. In view of Lemma~\ref{BS1:rsn} and Proposition~\ref{BS1:findex}, it suffices to
show that $K(X)$ is finitely generated and effectively computable.

Given $u \in F_A$, we have
\begin{displaymath}
u \notin R(X) \Leftrightarrow Xu \not\subseteq X
\Leftrightarrow Xu \cap (F_A \setminus X) \neq \emptyset
\Leftrightarrow u \in X^{-1}(F_A \setminus X),
\end{displaymath}
hence
\begin{displaymath}
R(X) = F_A \setminus(X^{-1}(F_A \setminus X))\, .
\end{displaymath}
It follows easily from the fact that the class of rational languages
is closed under reversal 
and morphisms, combined with Theorem~\ref{BS1:benois}(ii), that $X^{-1}
\in \Rat F_A$. Since $\Rat F_A$ is trivially closed under product, it
follows from Corollary~\ref{BS1:bencon} that $R(X)$ is rational and
effectively computable, and so is $K(X) = R(X) \cap (R(X))^{-1}$. By
Theorem~\ref{BS1:anisei}, the subgroup $K(X)$ is finitely generated
and the proof is complete.
\end{proof}

These results are related to the \Sakarovitch\
conjecture~\cite{Sakarovitch:1979}, which states that every rational subset
of $F_A$ must be either recognizable or
\emph{disjunctive}\index{disjunctive~rational~subset}: a subset $X$ of
a monoid $M$ is disjunctive if it has trivial syntactic congruence, or
equivalently, if any morphism $\varphi:M \to M'$ recognizing $X$ is
necessarily injective.

In the group case, it follows easily from the proof of the direct
implication of Lemma \ref{BS1:rsn} that the projection $G \to G/N$
recognizes $X \subseteq G$ if and only if $N \subseteq N(X)$. Thus $X$
is disjunctive if and only if $N(X)$ is the trivial subgroup.

The Sakarovitch conjecture was first proved in~\cite{MR1393764}, but
once again we follow the shorter alternative proof
from~\cite{MR2059028}:

\begin{theorem}[\Senizergues]
\label{BS1:saka}
A rational subset of $F_A$ is either recognizable or disjunctive.
\end{theorem}

\begin{proof}
Since the only subgroups of $\Z$ are the trivial
subgroup and finite index subgroups, we may assume that $\Card A > 1$.

Take $X \in \Rat F_A$. By the proof of Proposition~\ref{BS1:decrec},
the subgroup $K(X)$ is finitely generated. In view of
Lemma~\ref{BS1:rsn}, we may assume that $K(X)$ is not a finite index
subgroup. Thus ${\cal{S}}(K(X))$ is not complete by
Proposition~\ref{BS1:findex}.  Let $q_0$ denote the basepoint of
${\cal{S}}(K(X))$. Since ${\cal{S}}(K(X))$ is not complete, $q_0\cdot
u$ is undefined for some reduced word $u$.

Let $w$ be an arbitrary nonempty reduced word. We must show that $w
\notin N(X)$. Suppose otherwise. Since $u,w$ are reduced and $\Card A >
1$, there exist enough letters to make sure that there is some word $v
\in R_A$ such that $uvwv^{-1}u^{-1}$ is reduced. Now $w
\in N(X)$, hence $uvwv^{-1}u^{-1} \in N(X) \subseteq
K(X)$ by normality. Since $uvwv^{-1}u^{-1}$ is reduced, it follows
from Proposition~\ref{BS1:gwp} that $uvwv^{-1}u^{-1}$ labels a loop at $q_0$
in ${\cal{S}}(K(X))$, contradicting $q_0\cdot u$ being undefined. Thus $w
\notin N(X)$ and so $N(X) = 1$. Therefore $X$ is disjunctive as
required.
\end{proof}

\subsection{Beyond free groups}\label{BS1:beyondfg}
Let $\pi:F_A \twoheadrightarrow G$ be a morphism onto a group $G$.  We
consider the \emph{word problem submonoid}\index{word~problem!submonoid} of a group $G$, defined as
\begin{equation}\label{BS1:eq:wordproblem}
W_\pi(G) = (\pi\theta)^{-1}(\one).
\end{equation}

\begin{proposition}
\label{BS1:fincayley}
The language $W_{\pi}(G)$ is rational\index{word~problem!submonoid!rational} if and only if $G$ is
finite.\index{group!finite}
\end{proposition}

\begin{proof}
  If $G$ is finite, it is easy to check that $W_{\pi}(G)$ is rational
  by viewing the Cayley graph of $G$ (see~\S\ref{BS2:sec:1}) as an
  automaton. Conversely, if $W_{\pi}(G)$ is rational, then
  $\pi^{-1}(\one)$ is a finitely generated normal subgroup of $F_A$,
  either finite index or trivial by the proof of
  Theorem~\ref{BS1:saka}.  It is well known that the \emph{Dyck
    language}\index{Dyck~language}\index{language!Dyck} $D_A =
  \theta^{-1}(\one)$ is not rational if $\Card A > 0$, thus it follows
  easily that $\pi^{-1}(\one)$ has finite index and therefore $G$ must
  be finite.
\end{proof}

How about groups with context-free $W_\pi(G)$? A celebrated result by
\Muller\ and \Schupp~\cite{Muller&Schupp:1983}, with a contribution by
\Dunwoody~\cite{Dunwoody:1985}, relates them to \emph{virtually free
  groups}\index{group!virtually~free}\index{virtually~free~group}:
these are groups with a free subgroup of finite index.

As usual, we focus on the case of $G$ being finitely generated.  We
claim that $G$ has a \emph{normal} free subgroup $F_A$ of finite
index, with $A$ finite. Indeed, letting $F$ be a finite-index free
subgroup of $G$, it suffices to take $F' = \bigcap_{g\in G}\,
gFg^{-1}$. Since $F$ has finite index, so does $F'$, see the proof of
Lemma~\ref{BS1:rsn}. Taking a morphism $\pi:F_B \to G$ with $B$
finite, we get from Corollary~\ref{BS1:rankfis} that $\pi^{-1}(F')
\leq_{f.i.} F_B$ is finitely generated, so $F'$ is itself finitely
generated. Finally, $F'$ is a subgroup of $F$, so $F'$ is still free
by Theorem \ref{BS1:nielsen}, and we can write $F' = F_A$.

We may therefore decompose $G$ as a finite disjoint union of the form
\begin{equation}
\label{BS1:gottingen}
G = F_Ab_0 \cup F_Ab_1 \cup \dots \cup F_Ab_m,\qquad\text{with }b_0=1.
\end{equation}

\begin{theorem}[Muller \& Schupp]
\label{BS1:musc}
The language $W_{\pi}(G)$ is
context-free\index{word~problem!submonoid!context-free}\index{context-free!word~problem~submonoid}
if and only if $G$ is virtually free.\index{virtually~free~group}\index{group!virtually~free}
\end{theorem}

\begin{proof}[Sketch of proof]
If $G$ is virtually free, the rewriting system implicit
in~\eqref{BS1:gottingen}
provides a rational transduction between $W_{\pi}(G)$ and $D_A$.


  The converse implication can be proved by arguing geometrical
  properties of the Cayley graph of $G$ such as in
  Chapter~\ref{chapterBS2}; briefly said, one deduces from the
  context-freeness of $W_\pi(G)$ that the Cayley graph\index{Cayley~graph} of $G$ is
  close (more precisely, quasi-isometric\index{quasi-isometry}) to a tree.
\end{proof}

It follows that virtually free groups have decidable word problem.  In
Chapter~\ref{chapterBS2}, we shall discuss the word problem for more
general classes of groups using other techniques.

\Grunschlag\ proved that every rational (respectively recognizable)
subset of a virtually free group $G$ decomposed as in
(\ref{BS1:gottingen}) admits a decomposition as a
finite union $X_0b_0 \cup \dots \cup X_mb_m$, where the $X_i$ are
rational (respectively recognizable) subsets of $F_A$,
see~\cite{Grunschlag:1999}.  Thus basic results such as
Corollary~\ref{BS1:bencon} or Proposition~\ref{BS1:decrec} can be
extended to virtually free groups
(see~\cite{Grunschlag:1999,MR1920336}). Similar generalizations can be
obtained for free abelian
groups\index{group!free~abelian}\index{abelian~group!free} of finite
rank~\cite{MR1920336}.

The fact that the strong properties of Corollary~\ref{BS1:bencon} do
hold for both free groups and free abelian groups suggests
considering the case of graph groups (also known as free partially
abelian groups\index{group!free~partially~abelian} or right angled
Artin groups\index{group!right~angled
Artin}\index{group!Artin}), where we admit partial
commutation between letters.

An \emph{independence graph} is a finite undirected graph $(A,I)$ with
no loops, that is, $I$ is a symmetric anti-reflexive relation on
$A$. The \emph{graph group}\index{group!graph} $G(A,I)$ is the
quotient $F_A/\sim$, where
$\sim$ denotes the congruence generated by the relation
\begin{displaymath}
  \{ (ab,ba)\mid (a,b) \in I \}.
\end{displaymath}
On both extremes, we have $F_A = G(A,\emptyset)$ and the free abelian
group on $A$, which corresponds to the complete graph on $A$. These
turn out to be particular cases of
\emph{transitive forests}. We can say that $(A,I)$ is a transitive forest
if it has no induced subgraph of either of the following forms:
$$\begin{fsa}
\node (1) at (-3,1) {$\bullet$};
\node (2) at (-2,1) {$\bullet$};
\node (3) at (-3,0) {$\bullet$};
\node (4) at (-2,0) {$\bullet$};
\node (5) at (0,0) {$\bullet$};
\node (6) at (1,0) {$\bullet$};
\node (7) at (2,0) {$\bullet$};
\node (8) at (3,0) {$\bullet$};
\node at (-2.5,-0.5) {$C_4$};
\node at (1.5,-0.5) {$P_4$};
\draw[-] (1) edge node {} (2);
\draw[-] (2) edge node {} (4);
\draw[-] (4) edge node {} (3);
\draw[-] (3) edge node {} (1);
\draw[-] (5) edge node {} (6);
\draw[-] (6) edge node {} (7);
\draw[-] (7) edge node {} (8);
\end{fsa}$$
We recall that an induced subgraph of $(A,I)$ is formed by a subset of
vertices $A' \subseteq A$ and all the edges in $I$ connecting vertices
from $A'$.

The following difficult theorem, a group-theoretic version of a result
on trace monoids by \Aalbersberg\ and
\Hoogeboom~\cite{Aalbersberg&Hoogeboom:1989}, was proved
in~\cite{MR2422314}:

\begin{theorem}[\Lohrey\ \& \Steinberg]
\label{BS1:forests}
Let $(A,I)$ be an independence graph. Then $G(A,I)$ has decidable
rational subset membership problem\index{membership~problem!rational~subset} if and only if $(A,I)$ is a
transitive forest.
\end{theorem}

They also proved that these conditions are equivalent to decidability
of the membership problem for finitely generated submonoids. Such a
`bad' $G(A,I)$ gives an example of a finitely presented
group\index{group!finitely~presented} with a decidable generalized
word problem that does not have a decidable membership problem for
finitely generated submonoids.

It follows from Theorem \ref{BS1:forests} that any group containing a
direct product of two free monoids has undecidable rational subset
membership problem, a fact that can be directly deduced from the
undecidability of the \Post\ correspondence
problem\index{problem!Post~correspondence}.

Other positive results on rational subsets have been obtained for
graphs of groups\index{graphs~of~groups}, HNN
extensions\index{HNN~extension}\index{extension!HNN} and amalgamated
free
products\index{amalgamated~free~product}\index{free~product!amalgamated}
by \Kambites, \Silva\ and \Steinberg~\cite{MR2303197}, or \Lohrey\ and
\Senizergues~\cite{MR2394724}. Lohrey and Steinberg proved recently
that the rational subset membership problem is recursively equivalent
to the finitely generated submonoid membership problem for groups with
two or more ends~\cite{Lohrey&Steinberg:2009}.


With respect to closure under complement, \Lohrey\ and
\Senizergues~\cite{MR2394724} proved that the class of groups for
which the rational subsets form a boolean algebra is closed
under HNN extension and amalgamated products over finite groups.

On the negative side, \Bazhenova\ proved that rational subsets of
finitely generated nilpotent groups\index{group!nilpotent} do not form
a boolean algebra, unless the group is virtually
abelian~\cite{MR1803582}\index{group!virtually~abelian}.  Moreover,
\Romankov\ proved in \cite{Roman'kov:1999}, via a reduction from
Hilbert's 10th problem, that the rational subset membership problem is
undecidable for free nilpotent groups of any class $\geq 2$ of
sufficiently large rank.

Last but not least, we should mention that Stallings' construction was
successfully generalized to prove results on both graph
groups\index{Stallings'~construction!graph groups} (by \Kapovich,
\Miasnikov\ and \Weidmann~\cite{Kapovich&Weidmann&Miasnikov:2005}) and
amalgamated free products of finite
groups\index{Stallings'~construction!amalgamated~free~products~etc.}
(by \MarkusEpstein~\cite{Markus-Epstein:2007}).

\subsection{Rational solution sets and rational constraints}

In this final subsection we make a brief incursion in the brave new
world of rational constraints. Rational subsets provide group
theorists with two main assets:
\begin{itemize}
\item
A concept which generalizes finite generation for subgroups and is
much more fit to stand most induction procedures.
\item
A systematic way of looking for solutions of the \emph{right type} in
the context of equations of many sorts.
\end{itemize}
This second feature leads us to the notion of \emph{rational
  constraint}\index{rational~constraint}, when we restrict the set of
potential solutions to some rational subset. And there is a particular
combination of circumstances that can ensure the success of this strategy:
if $\Rat G$ is closed under intersection and we can prove that the
solution set of problem P is an effectively computable rational subset
of $G$, then we can solve problem P with any rational constraint.

An early example is the adaptation by \Margolis\ and \Meakin\ of
Rabin's language and \Rabin's tree theorem to free groups, where
first-order formulae provide rational solution
sets~\cite{MR1073775}. The logic language considered here is meant to
be applied to words, seen as models, and consists basically of unary
predicates that associate letters to positions in each word, as well
as a binary predicate for position ordering. \Margolis\ and \Meakin\
used this construction to solve problems in combinatorial inverse
semigroup\index{semigroup!inverse} theory~\cite{MR1073775}.

\Diekert, \Gutierrez\ and \Hagenah\ proved that the existential
theory\index{existential~theory~of~equations} of systems of equations
with rational constraints\index{equations!with~rational~constraints}
is solvable over a free group~\cite{MR2172984}. Working basically on a
free monoid with involution, and adapting \Plandowski's
approach~\cite{Plandowski:1999} in the process, they extended the
classical result of \Makanin~\cite{Makanin:1983} to include rational
constraints, with much lower complexity as well.

The proof of this deep result is well out of scope here, but its
potential applications are immense.
Group theorists are only starting to discover its full
strength.

The results in \cite{MR2394724} can be used to extend the existential
theory of equations with rational constraints to virtually free
groups\index{group!virtually~free}\index{virtually~free~group}, a
result that follows also from \Dahmani\ and \Guirardel's recent paper
on equations over hyperbolic
groups\index{word-hyperbolic~group}\index{group!word-hyperbolic} with
quasi-convex rational constraints \cite{Dahmani&Guirardel:2010}.
Equations over graph groups\index{group!graph} with a restricted class
of rational constraints were also successfully considered by \Diekert\
and \Lohrey~\cite{Diekert&Lohrey:2008}.



A somewhat exotic example of computation of a rational solution set
arises in the problem of determining which
automorphisms\index{automorphism!orbits} of $F_2$ (if any) carry a
given word into a given finitely generated subgroup. The full solution
set is recognized by a finite automaton; its vertices are themselves
structures named ``finite truncated
automata''\index{automaton!finite~truncated} \cite{Silva&Weil:2008b}.

\bibliographystyle{abbrv}
\addcontentsline{toc}{section}{References}
\begin{footnotesize}
  \bibliography{abbrevs,BS1}

\newcommand{\noopsort}[1]{} \newcommand{\singleletter}[1]{#1}
  \newcommand{\etal}{et al.}
\begin{thebibliography}{10}

\bibitem{Aalbersberg&Hoogeboom:1989}
I.~J. Aalbersberg and H.~J. Hoogeboom.
\newblock Characterizations of the decidability of some problems for regular
  trace languages.
\newblock {\em Math. Systems Theory}, 22:1--19, 1989.

\bibitem{CRAG}
Algebraic Cryptography Center.
\newblock {\em {CRAG -- the Cryptography and Groups Software Library}}, 2010.

\bibitem{MR1803582}
G.~A. Bazhenova.
\newblock On rational sets in finitely generated nilpotent groups.
\newblock {\em Algebra and Logic}, 39(4):215--223, 2000.
\newblock Translated from \emph{Algebra i Logika}, 39:379--394, 2000.

\bibitem{MR903667}
M.~Benois.
\newblock Descendants of regular language in a class of rewriting systems:
  algorithm and complexity of an automata construction.
\newblock In {\em Rewriting techniques and applications}, volume 256 of {\em
  Lecture Notes in Comput. Sci.}, pages 121--132. Springer-Verlag, 1987.

\bibitem{MR549481}
J.~Berstel.
\newblock {\em Transductions and context-free languages}.
\newblock B. G. Teubner, 1979.

\bibitem{Berstel&etal:2010}
J.~Berstel, C.~De~Felice, D.~Perrin, C.~Reutenauer, and G.~Rindone.
\newblock Bifix codes and sturmian words.
\newblock preprint, 2010.
\newblock \url{arXiv.org/pdf/1011.5369v2}.

\bibitem{Bestvina&Handel:1992}
M.~Bestvina and M.~Handel.
\newblock Train tracks and automorphisms of free groups.
\newblock {\em Ann. Math.}, 135:1--51, 1992.

\bibitem{MR1769781}
J.-C. Birget, S.~W. Margolis, J.~C. Meakin, and P.~Weil.
\newblock {PSPACE}-complete problems for subgroups of free groups and inverse
  finite automata.
\newblock {\em Theoret. Comput. Sci.}, 242(1-2):247--281, 2000.

\bibitem{Book&Otto:1993}
R.~V. Book and F.~Otto.
\newblock {\em String-Rewriting Systems}.
\newblock Springer-Verlag, 1993.

\bibitem{Dahmani&Guirardel:2010}
F.~Dahmani and V.~Guirardel.
\newblock Foliations for solving equations in groups: free, virtually free, and
  hyperbolic groups.
\newblock {\em J. Topology}, 3(2):343--404, 2010.

\bibitem{MR2172984}
V.~Diekert, C.~Gutierrez, and C.~Hagenah.
\newblock The existential theory of equations with rational constraints in free
  groups is {PSPACE}-complete.
\newblock {\em Inform. Comput.}, 202(2):105--140, 2005.

\bibitem{Diekert&Lohrey:2008}
V.~Diekert and M.~Lohrey.
\newblock Word equations over graph products.
\newblock {\em Internat. J. Algebra Comput.}, 18(3):493--533, 2008.

\bibitem{Dunwoody:1985}
M.~J. Dunwoody.
\newblock The accessibility of finitely presented groups.
\newblock {\em Invent. Math.}, 81(3):449--457, 1985.

\bibitem{TheGAPGroup:2004}
The GAP~Group.
\newblock {\em {GAP -- Groups, Algorithms, and Programming, Version 4.4.12}},
  2008.

\bibitem{MR695907}
S.~M. Gersten.
\newblock Intersections of finitely generated subgroups of free groups and
  resolutions of graphs.
\newblock {\em Inventiones Math.}, 71(3):567--591, 1983.

\bibitem{MR733696}
S.~M. Gersten.
\newblock On whitehead's algorithm.
\newblock {\em Bull. Amer. Math. Soc.}, 10(2):281--284, 1984.

\bibitem{MR847985}
R.~Z. Goldstein and E.~C. Turner.
\newblock Fixed subgroups of homomorphisms of free groups.
\newblock {\em Bull. Lond. Math. Soc.}, 18(5):468--470, 1986.

\bibitem{Grunschlag:1999}
Z.~Grunschlag.
\newblock {\em Algorithms in geometric group theory}.
\newblock PhD thesis, University of California at Berkeley, 1999.

\bibitem{MR2303197}
M.~Kambites, P.~V. Silva, and B.~Steinberg.
\newblock On the rational subset problem for groups.
\newblock {\em J. Algebra}, 309(2):622--639, 2007.

\bibitem{MR1882114}
I.~Kapovich and A.~Myasnikov.
\newblock Stallings foldings and subgroups of free groups.
\newblock {\em J. Algebra}, 248(2):608--668, 2002.

\bibitem{Kapovich&Weidmann&Miasnikov:2005}
I.~Kapovich, R.~Weidmann, and A.~Miasnikov.
\newblock Foldings, graphs of groups and the membership problem.
\newblock {\em Internat. J. Algebra Comput.}, 15(1):95--128, 2005.

\bibitem{MR2394724}
M.~Lohrey and G.~S{\'e}nizergues.
\newblock Rational subsets in {HNN}-extensions and amalgamated products.
\newblock {\em Internat. J. Algebra Comput.}, 18(1):111--163, 2008.

\bibitem{MR2422314}
M.~Lohrey and B.~Steinberg.
\newblock The submonoid and rational subset membership problems for graph
  groups.
\newblock {\em J. Algebra}, 320(2):728--755, 2008.

\bibitem{Lohrey&Steinberg:2009}
M.~Lohrey and B.~Steinberg.
\newblock Submonoids and rational subsets of groups with infinitely many ends.
\newblock To appear in \emph{J. Algebra}, 2009.

\bibitem{Makanin:1983}
G.~S. Makanin.
\newblock Equations in a free group.
\newblock {\em Math. USSR Izv.}, 21:483--546, 1983.
\newblock Translated from \emph{Izv. Akad. Nauk. SSR, Ser. Math.},
  46:1199--1273, 1983.

\bibitem{MR1214007}
S.~W. Margolis and J.~C. Meakin.
\newblock Free inverse monoids and graph immersions.
\newblock {\em Internat. J. Algebra Comput.}, 3(1):79--99, 1993.

\bibitem{MR1073775}
S.~W. Margolis and J.~C. Meakin.
\newblock Inverse monoids, trees and context-free languages.
\newblock {\em Trans. Amer. Math. Soc.}, 335(1):259--276, 1993.

\bibitem{Margolis&Sapir&Weil:2001}
S.~W. Margolis, M.~V. Sapir, and P.~Weil.
\newblock Closed subgroups in pro-{V} topologies and the extension problem for
  inverse automata.
\newblock {\em Internat. J. Algebra Comput.}, 11(4):405--446, 2001.

\bibitem{Markus-Epstein:2007}
L.~Markus-Epstein.
\newblock Stallings foldings and subgroups of amalgams of finite groups.
\newblock {\em Internat. J. Algebra Comput.}, 17(8):1493--1535, 2007.

\bibitem{Martino&Ventura:2004}
A.~Martino and E.~Ventura.
\newblock Fixed subgroups are compressed in free groups.
\newblock {\em Commun. Algebra}, 32(10):3921--3935, 2004.

\bibitem{Maslakova:2003}
O.~S. Maslakova.
\newblock The fixed point group of a free group automorphism.
\newblock {\em Algebra and Logic}, 42:237--265, 2003.
\newblock Translated from \emph{Algebra i Logika}, 42:422--472, 2003.

\bibitem{MR2395796}
A.~Miasnikov, E.~Ventura, and P.~Weil.
\newblock Algebraic extensions in free groups.
\newblock In {\em Geometric group theory}, Trends Math., pages 225--253.
  Birkh\"auser, 2007.

\bibitem{Muller&Schupp:1983}
D.~E. Muller and P.~E. Schupp.
\newblock Groups, the theory of ends, and context-free languages.
\newblock {\em J. Comput. System Sci.}, 26(3):295--310, 1983.

\bibitem{Neumann:1957}
H.~Neumann.
\newblock On the intersection of finitely generated free groups. addendum.
\newblock {\em Publ. Math. (Debrecen)}, 5:128, 1957.

\bibitem{Plandowski:1999}
W.~Plandowski.
\newblock Satisfiability of word equations with constants is in {PSPACE}.
\newblock In {\em Proc. 40th Ann. Symp. Found. Comput. Sci.}, pages 495--500.
  IEEE Press, 1999.

\bibitem{54.0603.01}
K.~Reidemeister.
\newblock {Fundamentalgruppe und \"Uberlagerungsr\"aume.}
\newblock {\em {J. Nachrichten G\"ottingen}}, pages 69--76, 1928.

\bibitem{Rhodes&Steinberg:2009}
J.~Rhodes and B.~Steinberg.
\newblock {\em The q-theory of finite semigroups}.
\newblock Springer-Verlag, 2009.

\bibitem{Ribes&Zalesskii:1993}
L.~Ribes and P.~A. Zalesskii.
\newblock On the profinite topology on a free group.
\newblock {\em Bull. Lond. Math. Soc.}, 25:37--43, 1993.

\bibitem{Ribes&Zalesskii:1994}
L.~Ribes and P.~A. Zalesskii.
\newblock The pro-{$p$} topology of a free group and algorithmic problems in
  semigroups.
\newblock {\em Internat. J. Algebra Comput.}, 4(3):359--374, 1994.

\bibitem{MR2378055}
A.~Roig, E.~Ventura, and P.~Weil.
\newblock On the complexity of the {W}hitehead minimization problem.
\newblock {\em Internat. J. Algebra Comput.}, 17(8):1611--1634, 2007.

\bibitem{Roman'kov:1999}
V.~Roman'kov.
\newblock On the occurrence problem for rational subsets of a group.
\newblock In V.~Roman'kov, editor, {\em International Conference on
  Combinatorial and Computational Methods in Mathematics}, pages 76--81, 1999.

\bibitem{Sakarovitch:1979}
J.~Sakarovitch.
\newblock {\em Syntaxe des langages de Chomsky, essai sur le d\'eterminisme}.
\newblock PhD thesis, Universit\'e Paris VII, 1979.

\bibitem{Sakarovitch:2003}
J.~Sakarovitch.
\newblock {\em El\'ements de th\'eorie des automates}.
\newblock Vuibert, 2003.

\bibitem{MR1057769}
G.~S{\'e}nizergues.
\newblock Some decision problems about controlled rewriting systems.
\newblock {\em Theoret. Comput. Sci.}, 71(3):281--346, 1990.

\bibitem{MR1393764}
G.~S{\'e}nizergues.
\newblock On the rational subsets of the free group.
\newblock {\em Acta Informatica}, 33(3):281--296, 1996.

\bibitem{MR0476875}
J.-P. Serre.
\newblock {\em Arbres, amalgames, ${\rm SL}\sb{2}$}.
\newblock Soci\'et\'e Math\'ematique de France, 1977.
\newblock Avec un sommaire anglais; r\'edig\'e avec la collaboration de Hyman
  Bass; \emph{Ast\'erisque} 46.

\bibitem{MR1920336}
P.~V. Silva.
\newblock Recognizable subsets of a group: finite extensions and the abelian
  case.
\newblock {\em Bull. European Assoc. Theor. Comput. Sci.}, 77:195--215, 2002.

\bibitem{MR2059028}
P.~V. Silva.
\newblock Free group languages: rational versus recognizable.
\newblock {\em RAIRO Inform. Th\'eor. App.}, 38(1):49--67, 2004.

\bibitem{Silva:2009}
P.~V. Silva.
\newblock Fixed points of endomorphisms over special confluent rewriting
  systems.
\newblock {\em Monatsh. Math.}, 161(4):417--447, 2010.

\bibitem{Silva&Weil:2008}
P.~V. Silva and P.~Weil.
\newblock On an algorithm to decide whether a free group is a free factor of
  another.
\newblock {\em RAIRO Inform. Th\'eor. App.}, 42:395--414, 2008.

\bibitem{Silva&Weil:2008b}
P.~V. Silva and P.~Weil.
\newblock Automorphic orbits in free groups: words versus subgroups.
\newblock {\em Internat. J. Algebra Comput.}, 20(4):561--590, 2010.

\bibitem{Silva&Weil:2008a}
P.~V. Silva and P.~Weil.
\newblock On finite-index extensions of subgroups of free groups.
\newblock {\em J. Group Theory}, 13(3):365--381, 2010.

\bibitem{MR1267733}
C.~C. Sims.
\newblock {\em Computation with finitely presented groups}.
\newblock Cambridge University Press, 1994.

\bibitem{MR695906}
J.~R. Stallings.
\newblock Topology of finite graphs.
\newblock {\em Inventiones Math.}, 71(3):551--565, 1983.

\bibitem{Takahasi:1951}
M.~Takahasi.
\newblock Note on chain conditions in free groups.
\newblock {\em Osaka J. Math.}, 3(2):221--225, 1951.

\bibitem{Touikan:2006}
W.~M. Touikan.
\newblock A fast algorithm for {S}tallings' folding process.
\newblock {\em Internat. J. Algebra Comput.}, 16(6):1031--1046, 2006.

\bibitem{Ventura:1997}
E.~Ventura.
\newblock On fixed subgroups of maximal rank.
\newblock {\em Commun. Algebra}, 25(10):3361--3375, 1997.

\bibitem{Ventura:2002}
E.~Ventura.
\newblock Fixed subgroups of free groups: a survey.
\newblock {\em Contemporary Math.}, 296:231--255, 2002.

\bibitem{Whitehead:1936}
J.~H.~C. Whitehead.
\newblock On equivalent sets of elements in a free group.
\newblock {\em Ann. of Math. (2)}, 37(4):782--800, 1936.

\end{thebibliography}
\end{footnotesize}

\newpage
\begin{abstract}
  This chapter is devoted to the study of rational subsets of groups,
  with particular emphasis on the automata-theoretic approach to
  finitely generated subgroups of free groups. Indeed, Stallings'
  construction, associating a finite inverse automaton with every such
  subgroup, inaugurated a complete rewriting of free group
  algorithmics, with connections to other fields such as topology or
  dynamics.

  Another important vector in the chapter is the fundamental Benois'
  Theorem, characterizing rational subsets of free groups. The theorem
  and its consequences really explain why language theory can be
  successfully applied to the study of free groups. Rational subsets
  of (free) groups can play a major role in proving statements
  (\emph{a priori} unrelated to the notion of rationality) by
  induction. The chapter also includes related results for more
  general classes of groups, such as virtually free groups or graph
  groups.
\end{abstract}

\printindex
\end{document}